\documentclass[showpacs,amsmath,amssymb,twocolumn,prl,superscriptaddress]{revtex4-1}

\usepackage{amsfonts,amsthm}
\usepackage{subfigure}
\usepackage{amssymb}
\usepackage{amscd}
\usepackage{amsmath}    
\usepackage{multirow}
\usepackage{graphicx}
\usepackage{dcolumn}
\usepackage{bm}
\usepackage{bbm}
\usepackage{cleveref}

\newtheorem{lemma}{Lemma}
\newtheorem{theorem}{Theorem}
\newcommand{\bra}[1]{\mbox{$\left\langle #1 \right|$}}
\newcommand{\ket}[1]{\mbox{$\left| #1 \right\rangle$}}

\begin{document}
\title{Detection loophole in quantum causality and its countermeasures}
\author{Zhu Cao}
\email{caozhu@ecust.edu.cn}
\address{Key Laboratory of Smart Manufacturing in Energy Chemical Process, Ministry of Education, East China University of Science and Technology, Shanghai 200237, China}

\begin{abstract}
Quantum causality violates classical intuitions of cause and effect and is a unique quantum feature different from other quantum phenomena such as entanglement and quantum nonlocality. 
In order to avoid the detection loophole in quantum causality, we initiate the study of the detection efficiency requirement for observing quantum causality.  We first show that previous classical causal inequalities require detection efficiency at least $95.97\%$ ($89.44\%$) to show violation with quantum (nonsignaling) correlations. Next we derive a classical causal inequality $I_{222}$ and show that it requires lower detection efficiency to be violated, $92.39\%$ for quantum correlations and $81.65\%$ for nonsignaling correlations, hence substantially reducing the requirement on detection. Then we extend this causal inequality to the case of multiple measurement settings and analyze the corresponding detection efficiency. After that, we show that previous quantum causal inequalities require detection efficiency at least $94.29\%$ to violate with nonsignaling correlations. We subsequently derive a  quantum causal bound $J_{222}$ that has a lower detection efficiency requirement of $91.02\%$ for violation with nonsignaling correlations. Our work paves the way towards an experimental demonstration of quantum causality and shows that causal inequalities significantly differ from Bell inequalities in terms of the detection efficiency requirement.
\end{abstract}

\vspace{2pc}

\maketitle

{\it Introduction.---}The identification of cause and effect is of paramount importance
for scientific discovery. However, traditional statistical methods are unable
to capture a causal relationship. In one aspect, from two correlated 
events alone, it is impossible to distinguish which is the cause and 
which is the effect. For example, for the two events ``rain'' and ``wet,'' one
cannot distinguish the causal relation ``rain causes wet'' from ``wet causes rain'' by statistical evidence alone.
In another aspect and perhaps more seriously, correlated events may have 
no causal relation at all, arising just from a common cause. For example,
 the events ``green'' and ``oxygen'' do not have a causal relationship 
 but are positively correlated because of a common cause ``plant.''
 To deal with the incapability of traditional statistical methods to capture causal
 relations, Pearl and his co-workers have developed a comprehensive
 mathematical framework of causality \cite{pearl2009causality,spirtes2000causation}. The tools of causality have
 since been applied to a wide variety of scientific fields \cite{glymour2001mind,morgan2015counterfactuals,shipley2016cause,peters2017elements}.

In the quantum regime, the theory of causality is dramatically different, which is referred to as \emph{quantum causality}.
One of the two fundamental cornerstones of classical causality, ``local realism'' and ``free will'' 
cannot both hold in the quantum regime due to the violation of Bell inequalities \cite{bell1964einstein}. 
This has spurred a large body of works that examine quantum nonlocality in various causal networks \cite{hall2020measurement,chaves2015unifying,wood2015lesson,wolfe2019inflation,hall2011relaxed,gallicchio2014testing,putz2014arbitrarily,friedman2019relaxed,hall2010local,handsteiner2017cosmic,rauch2018cosmic,abellan2018challenging,fritz2012beyond,chaves2016polynomial,rosset2016nonlinear,renou2019genuine,barrett2021cyclic}.
In a recent breakthrough, it is shown that quantum 
causality is a concept more general than quantum nonlocality \cite{PhysRevLett.125.230401}. Even in
causal setups where no Bell inequalities can be violated \cite{Henson_2014}, 
certain classical causal inequalities 
 can still be violated by quantum correlations \cite{PhysRevLett.125.230401}.
 Therefore, quantum causality illuminates the foundation of quantum theory from a new angle 
 and helps to deepen our understanding of quantum theory.

In quantum nonlocality, it is known that imperfect detection efficiency may destroy quantum nonlocality \cite{eberhard1993background}. Indeed, the \emph{detection loophole} caused by insufficient detection efficiency has been identified as one of the main bottlenecks for experimental observation of quantum nonlocality \cite{liu2018device}. This leads us to a natural question: Does imperfect detection efficiency also prevent the observation of quantum causality? And if so, what is the minimum detection efficiency requirement for observing quantum causality? These questions are very relevant from an experimental point of view, since realistic detectors do not have perfect detection efficiency. In addition, the experimental observation of quantum causality not only validates the theoretical prediction of quantum causality, but also facilitates the practical applications of quantum causality (to be discussed more later).

From a theoretical perspective, the detection efficiency requirement of quantum causality is also important due to two reasons. First, it may give a separation between quantum nonlocality and quantum causality, showing a second difference between these two concepts. Second, for a given bipartite quantum state, the minimum detection efficiency requirement for this state to violate causal inequalities can be used as a measure of its quantum causality, or more generally, its quantumness.

\begin{figure*}[htb]
\centering \includegraphics[width=16cm]{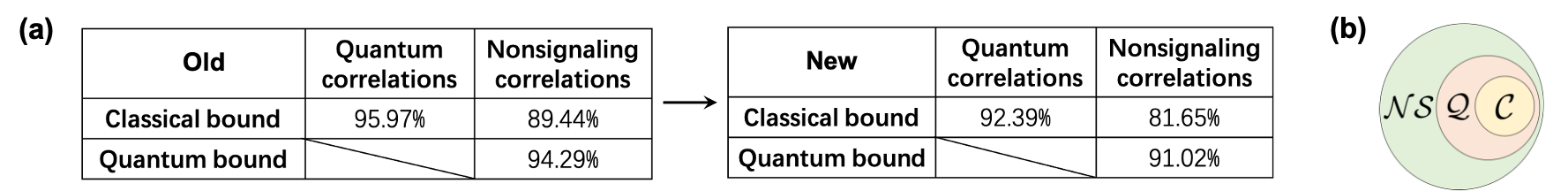}
\caption{(a) The detection efficiency requirements for quantum and nonsignaling correlations to violate classical and quantum causal bounds. Here, ``Old'' refers to the causal bounds in Ref.~\cite{PhysRevLett.125.230401}, and ``New'' refers to the causal bounds $I_{222}$ and $J_{222}$ derived in this Letter. (b) Illustration of the sets of classical correlations ($\mathcal{C}$), quantum correlations ($\mathcal{Q}$), and nonsignaling correlations ($\mathcal{NS}$).}
\label{fig:results}
\end{figure*} 

Despite its importance, the detection efficiency requirement for observing quantum causality has been elusive, partly due to the lack of tools in the emerging topic of quantum causality. In this work, we initiate the study of the detection efficiency  requirement for observing quantum causality, by borrowing tools from quantum nonlocality. As in the Bell scenario, here we assume that the detection efficiency can be manipulated by an adversary Eve. Note that if the efficiency loss is a trusted component, then simply dividing the probability of obtaining an outcome of a detector by the detector's detection efficiency suffices to correct the probability distortion caused by imperfect detection efficiency. The more difficult scenario is the case of untrusted efficiency loss, which we consider here.


We first analyze the detection efficiency requirement for 
showing quantum violation of the classical causal bound used in Ref.~\cite{PhysRevLett.125.230401}, and show it to be $95.97\%$. Next, in an extension of quantum theory, called the nonsignaling theory, we show that the corresponding detection efficiency needs to be at least $89.44\%$.
As these efficiency requirements are quite high, compared to the efficiency requirement $66.7\%$ for bipartite Bell violation with two measurement settings per party \cite{eberhard1993background}, we next aim to find an alternative classical causal bound that is less demanding on the detection efficiency. We manage to find a classical causal inequality $I_{222}$ which can be violated by quantum correlations and nonsignaling correlations with lower detection efficiency. 
Then we extend our  inequality $I_{222}$ to the case of $m$ measurement settings. After that, we examine the detection requirement for nonsignaling violation of the quantum causal bound used in Ref.~\cite{PhysRevLett.125.230401}, and show it to be $94.29\%$. Finally, we derive a  quantum causal bound that has a lower detection requirement of $91.02\%$ to violate. These results are summarized in Fig.~\ref{fig:results}(a).  All violations of causal bounds can be understood through the facts that classical correlations are contained in quantum correlations and quantum correlations are contained in nonsignaling correlations, as illustrated in Fig.~\ref{fig:results}(b).

From a theoretical perspective, our work gives an alternative proof that quantum causality and quantum nonlocality are two different physical concepts by showing that the detection efficiency requirement for observing quantum causality is bounded away from zero (recall the counterpart for quantum nonlocality can be arbitrarily close to zero \cite{PhysRevA.65.032121,PhysRevA.92.052104}). From an experimental perspective, our work provides guidance for the experimental observation of quantum causality and in addition reduces the efficiency requirement on detectors with our causal bounds $I_{222}$ and $J_{222}$. We hope our work will stimulate further theoretical and experimental research on quantum causality.


{\it Causal setup.---}As mentioned in the Introduction, whether a random variable $A$ implies 
another random variable $B$ cannot simply be deduced from the 
fact that they have a positive correlation. The positive correlation
may also result from a common cause of $A$ and $B$, denoted by $\Lambda$. Our 
goal is to separate the causal influence $A\to B$ from 
what can be explained by a potential common cause $\Lambda$. To this end, one
commonly used method is intervention \cite{pearl2009causality}, which sets $A$ to
 a value $a$ by force and examine the probability distribution of $B$, 
denoted by $\{ p[B=b|do(A=a)] \}_b$. To measure 
causality, we use the average causal effect (ACE) \cite{pearl2009causality}, which is defined as 
\begin{equation}
\textrm{ACE} =  \max\limits_{b,a,a'}[p(b| do(a))-p(b| do(a'))].
\end{equation}
This quantity measures the maximum change of $B$'s distribution when $A$'s value is altered. 

However, intervention cannot always be performed. For example, 
when a clinician wishes to examine the effect of smoking on people's health, he/she cannot just
force one group of people to smoke and the other group not to because this is unethical. In these
cases, only indirect estimation of ACE can be performed. One indirect method 
is to introduce a third variable $X$, called the instrumental variable, which 
is independent of the common cause $\Lambda$ and directly causes $A$ but not $B$. The relations between $\Lambda$, $A$, $B$, and $X$ 
are illustrated in the left panel of Fig.~\ref{fig:causalnodes}.
By these relations, we can express the probability of observing $A=a$ and $B=b$ given $X=x$ 
as
\begin{equation}
\label{eq:probclassical}
p(a,b|x) = \sum\limits_\lambda p_A(a|x,\lambda) p_B(b| a,\lambda)p(\lambda),
\end{equation}
and the probability of observing $B=b$ when fixing $A=a$ by force as
\begin{equation}
p(b|do(a)) = \sum\limits_\lambda  p_B(b| a,\lambda)p(\lambda).
\end{equation}
It has been shown that for this setup, no Bell inequalities can be violated \cite{Henson_2014}.

\begin{figure}[htb]
\centering \includegraphics[width=6cm]{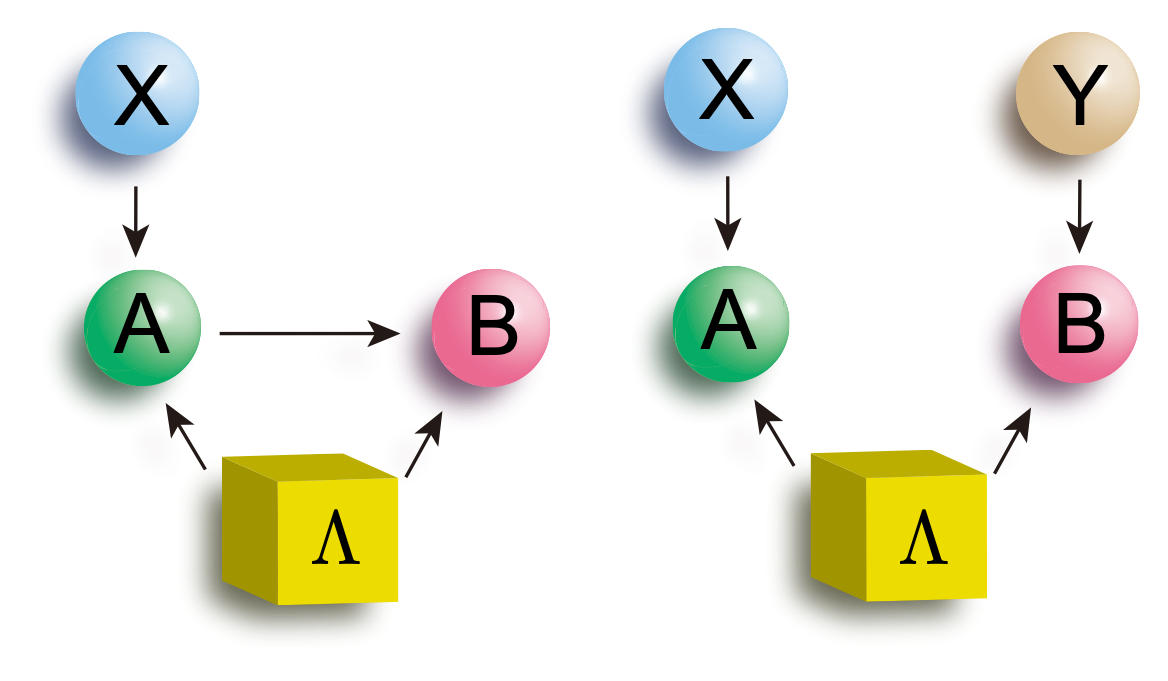}
\caption{Illustration of the causal setup. Left: The random variable $\Lambda$ is the common cause of
two random variables $A$ and $B$. The random variable $X$ has a causal influence on $A$, and $A$ has a 
causal influence on $B$. Right: Mapping of the causal scenario to the Bell scenario. In the Bell scenario, the random variable $B$ is causally influenced by a random variable $Y$ instead of $A$. Identification of $A$ and $Y$ recovers the causal scenario.}
\label{fig:causalnodes}
\end{figure} 

In the quantum setting, the common cause for $A$ and $B$ is replaced by 
a quantum state $\rho_{AB}$. The random variables $A$ and $B$
perform local measurements $M^x_a$ and $N^a_b$, respectively, on
 $\rho_{AB}$ to obtain their outcomes. Here, $M^x_a$ 
and $N^a_b$ use $x$ and $a$, respectively, to choose their 
measurement settings and output $a$ and $b$, respectively.
Hence, in the quantum setting, we can express the probability of observing $A=a$ and $B=b$ given $X=x$ 
as
\begin{equation}
\label{eq:probquantum}
p(a,b|x) = \textrm{tr}[(M^x_a \otimes N^a_b)\rho_{AB}],
\end{equation}
and the probability of observing $B=b$ when fixing $A=a$ by force as
\begin{equation}
p(b|do(a)) = \textrm{tr}[(\mathbbm{1} \otimes N^a_b)\rho_{AB}].
\end{equation}

We can further extend the quantum setting to the nonsignaling theory, 
where correlations are only constrained by the nonsignaling condition. To explain 
the nonsignaling condition in the causal scenario, we first map the causal
probability distribution to a Bell probability distribution as
\begin{equation}
\begin{aligned}
p(a,b| x) &= p_{Bell} (a,b|x,a),  \quad \forall  a,b,x,  \\
p(b|do(a)) &= \sum\limits_{a'} p_{Bell}(a',b|x,a),  \quad \forall  a,b,x. 
\end{aligned}
\end{equation}
where $p_{Bell}$ ($p$) denotes the Bell (causal) probability distribution and
$p_{Bell}(a,b|x,y)$ denotes a bipartite Bell scenario with inputs $x,y$ and outputs $a,b$.

The mapped Bell scenario is illustrated in the right panel of Fig.~\ref{fig:causalnodes}.
The nonsignaling condition is then
\begin{eqnarray*}
\sum\limits_{a'} p_{Bell}(a',b|x_1,y) = \sum\limits_{a'} p_{Bell}(a',b|x_2,y), \;   \forall  y,b,x_1,x_2,  \nonumber \\
\sum\limits_{b'} p_{Bell}(a,b'|x,y_1) = \sum\limits_{b'} p_{Bell}(a,b'|x,y_2),  \; \forall  x,a,y_1,y_2,
\end{eqnarray*}
where $x_1,x_2,x,y_1,y_2,y$ are the inputs and $a,b,a',b'$ are the outputs of the Bell test.

{\it Model of detection efficiency.---}For trusted detectors, we assume their detection efficiency 
is a fixed value independent of the state to be detected. This assumption is physically justified as the detection efficiency of polarization-encoded qubits by photon detectors is independent of the polarization \cite{cao2016source}. 
In later analysis, we will show that even when restricting to the use of trusted detectors, quantum correlations are still able to violate classical causal bounds.

For untrusted detectors, their detection efficiency is not assumed to be independent of the input state. However, these untrusted detectors are enforced to have the same overall probability of detection as trusted detectors, namely, 
\begin{equation}
\sum\limits_{a,b=0}^1 p(a,b|x) = \eta^2, \quad \forall x.
\end{equation}
where $\eta$ is the detection efficiency of trusted detectors. If this condition is violated, we can easily detect the malicious behavior of the detectors and terminate the test of quantum causality. In later analysis, we will show that even with untrusted detectors, classical causal bounds are still obeyed by classical correlations. The proofs are quite different from the case of perfect detection efficiency, which is the only case that previous work has considered \cite{balke1997bounds,PhysRevLett.125.230401}.

\label{sec:results}

{\it Detection efficiency of quantum violation of classical bound.---}With the problem setup in place, we now consider the following classical causal lower bound which first 
appeared in Ref.~\cite{balke1997bounds} and was first compared against quantum correlations in Ref.~\cite{PhysRevLett.125.230401}:
\begin{equation}
\label{eq:CACE}
ACE \ge 2 p(0,0 | 0 ) + p(1,1 |0) + p(0,1|1) + p( 1,1|1) -2,
\end{equation}
where $p(a,b|x)$ is defined as in Eqs.~\eqref{eq:probclassical} and \eqref{eq:probquantum} for classical correlations and quantum correlations, respectively.  An illustration of this bound is shown in Fig.~\ref{fig:illusI222}. 
\begin{figure}[htb]
\centering \includegraphics[width=5cm]{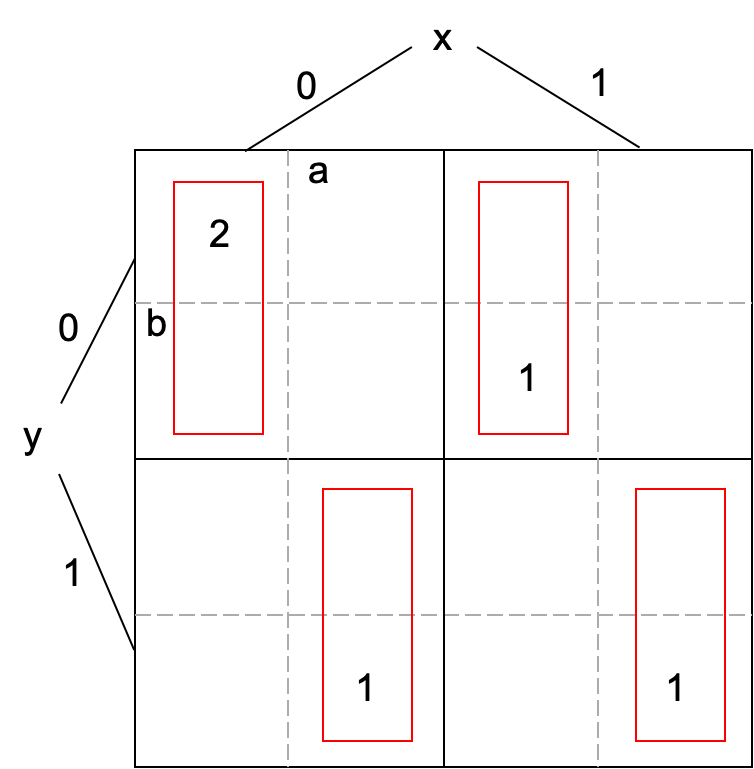}
\caption{Illustration of the classical causality lower bound in Ref.~\cite{PhysRevLett.125.230401}. Because of the condition $a=y$, only the entries in the red boxes can appear in a causal inequality.} 
\label{fig:illusI222}
\end{figure}

We start by examining the required detection efficiency for quantum correlations to violate this bound.
The result is summarized in the following theorem (see Supplemental Material Sec. I for its proof).

\begin{theorem}
\label{thm:quantum2}
When the detection efficiency $\eta$ is larger than $ \sqrt{2/(5-2\sqrt{2})}  \approx 95.97\%$, quantum correlations can violate the classical causal bound Eq.~\eqref{eq:CACE}, the maximal violation of which is $3-2\sqrt{2}$.
\end{theorem}

From this result, we have two observations. First, imperfect detectors suffice for showing quantum violation of classical causal bounds as long as their efficiency exceeds the threshold specified in the theorem. Second, the detection efficiency requirement of the causal inequality Eq.~\eqref{eq:CACE} to be violated by quantum correlations is much higher than Clauser-Horne-Shimony-Holt inequality \cite{clauser1969proposed}, although they have a similar form.

{\it Detection efficiency of nonsignaling violation of classical bound.---}Next, we consider the efficiency bound of the causal inequality Eq.~\eqref{eq:CACE} for the nonsignaling theory.
The  result is given by the following theorem (see Supplemental Material Sec. II for its  proof).

\begin{theorem}
\label{thm:nosignal2}
When the detection efficiency $\eta$ is larger than $ \sqrt{4/5} \approx 89.44\%$, nonsignaling correlations can violate the classical causal bound Eq.~\eqref{eq:CACE}, the maximal violation of which is $1/2$.
\end{theorem}

From this theorem, we observe two facts. Firstly, nonsignaling correlations are able to violate the classical bound more than which is capable by quantum correlations, as $1/2 >3-2\sqrt{2}$. Secondly, the efficiency requirement for nonsignaling correlations to violate the classical bound is lower than the one for quantum correlations. This suggests that in the search for classical bounds with a lower detection efficiency requirement, it is useful to start by examining the detection efficiency requirement for nonsignaling violation of these classical bounds, which we turn to next.

{\it Causal inequality with lower efficiency requirement for nonsignaling violation.---}As can be seen from the proof of Theorem \ref{thm:nosignal2}, the main restriction of the detection efficiency is the constant term $-2$ in the causal inequality Eq.~\eqref{eq:CACE}. In order to reduce the detection efficiency requirement, it suffices to shrink this constant term. To this end,
we consider the following quantity, 
\begin{equation*}
I_{222} = 2 p(0,0 | 0 ) + p(1,1 |0) + p(0,1|1) + p( 1,1|1) - 1 -\eta^2,
\end{equation*}
and a causal inequality
\begin{equation}
\label{eq:CACEnew}
ACE \ge I_{222}.
\end{equation}
The three subscripts of $I_{222}$ represent the range of values for $X$, $A$, and $B$, respectively.
With this causal inequality at hand, we are able to show the following theorem (see Supplemental Material Sec. III for its  proof).

\begin{theorem}
The detection efficiency needs to be at least $ \sqrt{2/3} \approx 81.65\%$ for nonsignaling correlations to violate the classical causal bound Eq.~\eqref{eq:CACEnew}, the maximal violation of which is $1/2$.
\end{theorem}

The result of this theorem hints that the detection requirement of the quantum case can be reduced as well, which we turn to next.

{\it Causal inequality with lower efficiency requirement for showing quantum violation.---}To reduce the detection efficiency requirement for showing quantum violation of classical causal bounds, we consider the causal inequality Eq.~\eqref{eq:CACEnew} in the quantum scenario. The result is summarized in the following theorem (see Supplemental Material Sec. IV for its  proof).

\begin{theorem} 
The detection efficiency needs to be at least  $ \sqrt{1/(4-2\sqrt{2})}  \approx 92.39\%$ for quantum correlations to violate the classical causal bound Eq.~\eqref{eq:CACEnew}, the maximal violation of which is $3-2\sqrt{2}$.
\end{theorem}

The fact that quantum violation of causality inequalities requires near unit detection efficiency stands in stark contrast to the Bell scenario. In the Bell case, the detection efficiency requirement of bipartite Bell violation with quantum correlations can approach 0 with a large quantum system dimension and an exponential number of measurement settings with respect to the quantum system dimension \cite{PhysRevA.65.032121}. Alternatively, the upper bound on the quantum efficiency requirement for Bell violation for the $n$-partite case with $n$ measurement settings per party is $2/(n+1)$ \cite{PhysRevA.92.052104}, which also approaches 0 as $n$ goes to infinity. This large discrepancy of the detection efficiency requirement between the causal scenario and the Bell scenario shows that quantum causality is different from quantum nonlocality. The main obstacle for reducing the detection efficiency requirement for showing quantum violation of causal inequalities lies in the fact that a causal inequality can only contain the terms of the form $p(a,b|x)=p_{Bell}(a,b|x,a)$ while a Bell inequality can contain all terms of the form $p_{Bell}(a,b|x,y)$.

{\it  Nonsignaling violation of causal bounds with more than two measurement settings.---}Then we generalize the causal inequality to multiple measurement settings. 
In more detail, suppose $a,b\in \{0, 1\}$ and $x \in \{0,\dots, M-1\}$. We consider the following quantity:
\begin{equation}
\begin{aligned}
I_{M22} = \frac{M}{M-1} p(0,0| 0) + \frac{1}{M-1} p(0,1|1)  \\
 +\frac{1}{M-1} \sum\limits_{x=0}^{M-1} p(1,1|x) -1 -\frac{\eta^2}{M-1} .  
\end{aligned}
\end{equation}
An illustration of $I_{322}$ and $I_{422}$ is shown in Fig.~\ref{fig:Im22}.
 \begin{figure}[htb]
\centering \includegraphics[width=8.5cm]{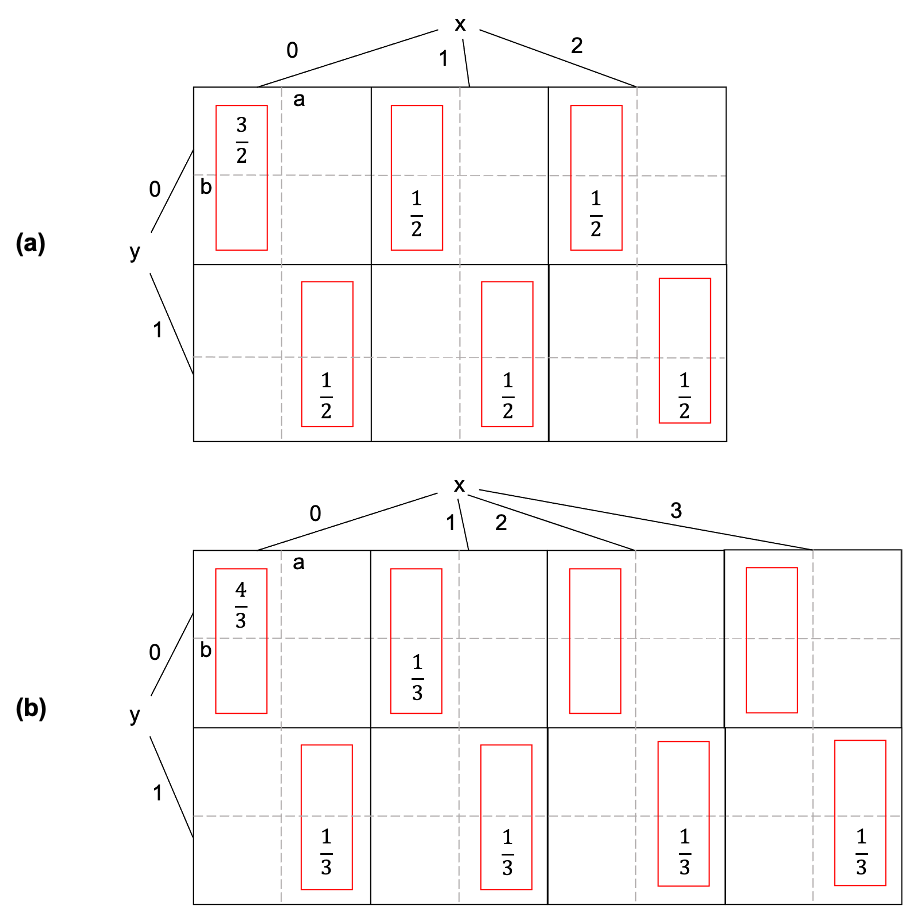}
\caption{Illustration of classical causal bounds with $M$ measurement settings for (a) $M=3$ and (b) $M=4$.} 
\label{fig:Im22}
\end{figure} 
The causal inequality we examine is 
\begin{equation}
\label{eq:multi}
ACE \ge I_{M22}.
\end{equation} 

The detection efficiency requirement of this causal inequality is summarized in the following theorem (see Supplemental Material Sec. V for its  proof).

\begin{theorem}
When system $A$ has $M$ measurement settings, as long as the detection efficiency exceeds $\sqrt{(2M-2)/(2M-1)}$, nonsignaling correlations can violate the classical bound Eq.~\eqref{eq:multi}, the maximal violation of which is $1/(2M-2)$.
\end{theorem}

It can be seen that an increase in the measurement settings leads to a smaller violation of the causal inequality and a higher detection requirement. This fact also differs from the Bell scenario where more measurement settings will result in a larger violation and a lower detection requirement \cite{PhysRevA.94.042126}.

{\it Detection efficiency of nonsignaling correlations to violate quantum causal bound.---}It has been postulated that quantum theory cannot fully explain phenomena in ultra-high-density objects such as black holes
or in ultra-small scales such as a Planck length, and an extension of quantum theory is needed in these regimes \cite{Preskill1992Do}.
To witness such a theory, we consider the task of violating quantum causal bounds.

As shown in Ref.~\cite{PhysRevLett.125.230401}, quantum correlations obey a causal lower bound
\begin{equation}
\label{eq:QACE}
ACE \ge \sum\limits_{x=0,1} [p(0,0|x) + p(1,1|x)] - \xi - 1, 
\end{equation}
where 
\begin{equation*}
\xi = \min\limits_{\pm} \sqrt{  \prod\limits_{a=0,1} \{ 1\pm \sum\limits_{x=0,1} (-1)^x[ p(a,0|x) -  p(a,1|x) ]  \}   } .
\end{equation*}
To witness a theory more general than quantum theory such as the nonsignaling theory, we consider its violation of this bound.
The result is summarized in the following theorem (see Supplemental Material Sec. VI for its  proof).

\begin{theorem}
When the detection efficiency $\eta$ is larger than $94.29\%$, nonsignaling correlations can violate the quantum bound Eq.~\eqref{eq:QACE}, the maximal violation of which is $1/2$.
\end{theorem}

This result has two implications. First, it shows that indeed the nonsignaling theory can be witnessed through its violation of a quantum causal bound.
Second, even with imperfect detection efficiency, the quantum causal bound can still be violated.

{\it Quantum causal bound with a lower efficiency requirement.---}To lower the detection efficiency requirement for nonsignaling violation of quantum
causal bounds, we consider the following quantity
\begin{eqnarray*}
J_{222} &=&[p(0,0|0) + p(1,1|0)-p(0,1|1) -p(1,0|1) - \xi ]/\eta,  \\
\xi &=& \min\limits_{\pm} \sqrt{  \prod\limits_{a=0,1} \{ 1 \pm \sum\limits_{x=0,1} (-1)^x[ p(a,0|x) -  p(a,1|x) ]  \}   } , \nonumber
\end{eqnarray*}
and the following inequality
\begin{equation}
\label{eq:QACE2}
ACE \ge J_{222}.
\end{equation}
Here, the subscripts of $J_{222}$ have similar meanings with that of $I_{222}$.
The property of the inequality Eq.~\eqref{eq:QACE2} is given in the following theorem (see Supplemental Material Sec. VII for its  proof).

\begin{theorem}
When the detection efficiency $\eta$ is larger than $91.02\%$, nonsignaling correlations can violate the quantum causal bound Eq.~\eqref{eq:QACE2}, the maximal violation of which is $1/2$.
\end{theorem}

This result shows that a higher detection efficiency 
is needed for nonsignaling correlations to violate a quantum causal bound than a classical causal bound.

{\it Applications.---}Similar to quantum nonlocality, quantum causality without loopholes has applications to a wide range of device-independent protocols. In particular, even in scenarios where no Bell inequalities can be violated (i.e., the resource of quantum nonlocality is not available), device independence can still be guaranteed by exploiting quantum causality. Consider the following application of quantum causality to device-independent quantum random number generation. When a classical causal bound, such as Eq.~\eqref{eq:CACE}, is maximally violated, the outcome $b$ of the variable $B$ must be generated by a quantum process. Hence, the value $b$ is genuinely random by Born's rule. We leave a more detailed analysis of this protocol, including its random number generation rate as a function of the value of violation, to future work. As another example, we apply quantum causality to  device-independent quantum key distribution (DI-QKD). As the maximal violation of a causal inequality certifies the presence of an entangled state between $A$ and $B$, an entanglement-based quantum key distribution procedure combined with occasionally testing the violation of the causal inequality suffices to achieve DI-QKD. We also leave a more detailed analysis of this protocol, including its key rate as a function of the violation, to future work.

{\it Conclusion.---}Quantum causality provides a distinct new lens to understand the foundation of quantum theory.
In this work, we have analyzed the detection efficiency requirements for quantum correlations and 
nonsignaling correlations to violate the causal inequalities. We have shown that quantum 
violation requires detection efficiency at least $95.97\%$ and nonsignaling violation requires detection efficiency at least $89.44\%$ in previous causal bounds. 
To lower the efficiency requirement, we have proposed a causal inequality $I_{222}$. 
We have shown that both nonsignaling correlations and quantum correlations can violate this  inequality
with lower detection efficiency. In addition, we have generalized this  causal inequality to 
multiple measurement settings. We have examined the detection requirement for nonsignaling violation of the quantum causal bound derived in Ref.~\cite{PhysRevLett.125.230401}. We have also proposed a  quantum causal bound $J_{222}$ that is less demanding on detection efficiency for showing nonsignaling violation.

Our work advances the emerging field of quantum causality both experimentally and theoretically. From an experimental perspective, our work provides guidance for the selection of experimental parameters and hence paves the way towards experimental observation of quantum causality. From a theoretical perspective, our work shows that quantum causality differs from quantum nonlocality in terms of the detection efficiency requirement (one is bounded away from zero, while the other can approach 0). This deepens our understanding of quantum causality and consequently the foundation of quantum theory.

Our work opens a few prospective avenues for future research. On the experimental side,
it would be interesting to experimentally demonstrate the quantum violation of our  causal bounds.
On the theoretical side, extending the quantum violation of causal 
inequalities to the case of more than two parties and examining the corresponding 
detection efficiency requirement are interesting open problems.

\begin{acknowledgements}
This work was supported by National Natural Science Foundation of China (Basic Science Center Program: 61988101), International (Regional) Cooperation and Exchange Project (61720106008), National Natural Science Fund for Distinguished Young Scholars (61725301),  Natural Science Foundation of Shanghai (21ZR1415800), Shanghai Sailing Program, and the startup fund from East China University of Science and Technology (SLH00212004).
\end{acknowledgements}

\bibliographystyle{apsrev4-1}

\bibliography{BibliCausal}

\begin{thebibliography}{35}%
\makeatletter
\providecommand \@ifxundefined [1]{%
 \@ifx{#1\undefined}
}%
\providecommand \@ifnum [1]{%
 \ifnum #1\expandafter \@firstoftwo
 \else \expandafter \@secondoftwo
 \fi
}%
\providecommand \@ifx [1]{%
 \ifx #1\expandafter \@firstoftwo
 \else \expandafter \@secondoftwo
 \fi
}%
\providecommand \natexlab [1]{#1}%
\providecommand \enquote  [1]{``#1''}%
\providecommand \bibnamefont  [1]{#1}%
\providecommand \bibfnamefont [1]{#1}%
\providecommand \citenamefont [1]{#1}%
\providecommand \href@noop [0]{\@secondoftwo}%
\providecommand \href [0]{\begingroup \@sanitize@url \@href}%
\providecommand \@href[1]{\@@startlink{#1}\@@href}%
\providecommand \@@href[1]{\endgroup#1\@@endlink}%
\providecommand \@sanitize@url [0]{\catcode `\\12\catcode `\$12\catcode
  `\&12\catcode `\#12\catcode `\^12\catcode `\_12\catcode `\%12\relax}%
\providecommand \@@startlink[1]{}%
\providecommand \@@endlink[0]{}%
\providecommand \url  [0]{\begingroup\@sanitize@url \@url }%
\providecommand \@url [1]{\endgroup\@href {#1}{\urlprefix }}%
\providecommand \urlprefix  [0]{URL }%
\providecommand \Eprint [0]{\href }%
\providecommand \doibase [0]{http://dx.doi.org/}%
\providecommand \selectlanguage [0]{\@gobble}%
\providecommand \bibinfo  [0]{\@secondoftwo}%
\providecommand \bibfield  [0]{\@secondoftwo}%
\providecommand \translation [1]{[#1]}%
\providecommand \BibitemOpen [0]{}%
\providecommand \bibitemStop [0]{}%
\providecommand \bibitemNoStop [0]{.\EOS\space}%
\providecommand \EOS [0]{\spacefactor3000\relax}%
\providecommand \BibitemShut  [1]{\csname bibitem#1\endcsname}%
\let\auto@bib@innerbib\@empty
\bibitem [{\citenamefont {Pearl}(2009)}]{pearl2009causality}%
  \BibitemOpen
  \bibfield  {author} {\bibinfo {author} {\bibfnamefont {J.}~\bibnamefont
  {Pearl}},\ }\href@noop {} {\emph {\bibinfo {title} {Causality}}}\ (\bibinfo
  {publisher} {Cambridge university press},\ \bibinfo {year}
  {2009})\BibitemShut {NoStop}%
\bibitem [{\citenamefont {Spirtes}\ \emph {et~al.}(2000)\citenamefont
  {Spirtes}, \citenamefont {Glymour}, \citenamefont {Scheines},\ and\
  \citenamefont {Heckerman}}]{spirtes2000causation}%
  \BibitemOpen
  \bibfield  {author} {\bibinfo {author} {\bibfnamefont {P.}~\bibnamefont
  {Spirtes}}, \bibinfo {author} {\bibfnamefont {C.~N.}\ \bibnamefont
  {Glymour}}, \bibinfo {author} {\bibfnamefont {R.}~\bibnamefont {Scheines}}, \
  and\ \bibinfo {author} {\bibfnamefont {D.}~\bibnamefont {Heckerman}},\
  }\href@noop {} {\emph {\bibinfo {title} {Causation, prediction, and
  search}}}\ (\bibinfo  {publisher} {MIT press},\ \bibinfo {year}
  {2000})\BibitemShut {NoStop}%
\bibitem [{\citenamefont {Glymour}(2001)}]{glymour2001mind}%
  \BibitemOpen
  \bibfield  {author} {\bibinfo {author} {\bibfnamefont {C.~N.}\ \bibnamefont
  {Glymour}},\ }\href@noop {} {\emph {\bibinfo {title} {The mind's arrows:
  Bayes nets and graphical causal models in psychology}}}\ (\bibinfo
  {publisher} {MIT press},\ \bibinfo {year} {2001})\BibitemShut {NoStop}%
\bibitem [{\citenamefont {Morgan}\ and\ \citenamefont
  {Winship}(2015)}]{morgan2015counterfactuals}%
  \BibitemOpen
  \bibfield  {author} {\bibinfo {author} {\bibfnamefont {S.~L.}\ \bibnamefont
  {Morgan}}\ and\ \bibinfo {author} {\bibfnamefont {C.}~\bibnamefont
  {Winship}},\ }\href@noop {} {\emph {\bibinfo {title} {Counterfactuals and
  causal inference}}}\ (\bibinfo  {publisher} {Cambridge University Press},\
  \bibinfo {year} {2015})\BibitemShut {NoStop}%
\bibitem [{\citenamefont {Shipley}(2016)}]{shipley2016cause}%
  \BibitemOpen
  \bibfield  {author} {\bibinfo {author} {\bibfnamefont {B.}~\bibnamefont
  {Shipley}},\ }\href@noop {} {\emph {\bibinfo {title} {Cause and correlation
  in biology: a user's guide to path analysis, structural equations and causal
  inference with R}}}\ (\bibinfo  {publisher} {Cambridge University Press},\
  \bibinfo {year} {2016})\BibitemShut {NoStop}%
\bibitem [{\citenamefont {Peters}\ \emph {et~al.}(2017)\citenamefont {Peters},
  \citenamefont {Janzing},\ and\ \citenamefont
  {Sch{\"o}lkopf}}]{peters2017elements}%
  \BibitemOpen
  \bibfield  {author} {\bibinfo {author} {\bibfnamefont {J.}~\bibnamefont
  {Peters}}, \bibinfo {author} {\bibfnamefont {D.}~\bibnamefont {Janzing}}, \
  and\ \bibinfo {author} {\bibfnamefont {B.}~\bibnamefont {Sch{\"o}lkopf}},\
  }\href@noop {} {\emph {\bibinfo {title} {Elements of causal inference}}}\
  (\bibinfo  {publisher} {The MIT Press},\ \bibinfo {year} {2017})\BibitemShut
  {NoStop}%
\bibitem [{\citenamefont {Bell}(1964)}]{bell1964einstein}%
  \BibitemOpen
  \bibfield  {author} {\bibinfo {author} {\bibfnamefont {J.~S.}\ \bibnamefont
  {Bell}},\ }\href@noop {} {\bibfield  {journal} {\bibinfo  {journal} {Physics
  Physique Fizika}\ }\textbf {\bibinfo {volume} {1}},\ \bibinfo {pages} {195}
  (\bibinfo {year} {1964})}\BibitemShut {NoStop}%
\bibitem [{\citenamefont {Hall}\ and\ \citenamefont
  {Branciard}(2020)}]{hall2020measurement}%
  \BibitemOpen
  \bibfield  {author} {\bibinfo {author} {\bibfnamefont {M.~J.~W.}\
  \bibnamefont {Hall}}\ and\ \bibinfo {author} {\bibfnamefont {C.}~\bibnamefont
  {Branciard}},\ }\href@noop {} {\bibfield  {journal} {\bibinfo  {journal}
  {Physical Review A}\ }\textbf {\bibinfo {volume} {102}},\ \bibinfo {pages}
  {052228} (\bibinfo {year} {2020})}\BibitemShut {NoStop}%
\bibitem [{\citenamefont {Chaves}\ \emph {et~al.}(2015)\citenamefont {Chaves},
  \citenamefont {Kueng}, \citenamefont {Brask},\ and\ \citenamefont
  {Gross}}]{chaves2015unifying}%
  \BibitemOpen
  \bibfield  {author} {\bibinfo {author} {\bibfnamefont {R.}~\bibnamefont
  {Chaves}}, \bibinfo {author} {\bibfnamefont {R.}~\bibnamefont {Kueng}},
  \bibinfo {author} {\bibfnamefont {J.~B.}\ \bibnamefont {Brask}}, \ and\
  \bibinfo {author} {\bibfnamefont {D.}~\bibnamefont {Gross}},\ }\href@noop {}
  {\bibfield  {journal} {\bibinfo  {journal} {Physical Review Letters}\
  }\textbf {\bibinfo {volume} {114}},\ \bibinfo {pages} {140403} (\bibinfo
  {year} {2015})}\BibitemShut {NoStop}%
\bibitem [{\citenamefont {Wood}\ and\ \citenamefont
  {Spekkens}(2015)}]{wood2015lesson}%
  \BibitemOpen
  \bibfield  {author} {\bibinfo {author} {\bibfnamefont {C.~J.}\ \bibnamefont
  {Wood}}\ and\ \bibinfo {author} {\bibfnamefont {R.~W.}\ \bibnamefont
  {Spekkens}},\ }\href@noop {} {\bibfield  {journal} {\bibinfo  {journal} {New
  Journal of Physics}\ }\textbf {\bibinfo {volume} {17}},\ \bibinfo {pages}
  {033002} (\bibinfo {year} {2015})}\BibitemShut {NoStop}%
\bibitem [{\citenamefont {Wolfe}\ \emph {et~al.}(2019)\citenamefont {Wolfe},
  \citenamefont {Spekkens},\ and\ \citenamefont {Fritz}}]{wolfe2019inflation}%
  \BibitemOpen
  \bibfield  {author} {\bibinfo {author} {\bibfnamefont {E.}~\bibnamefont
  {Wolfe}}, \bibinfo {author} {\bibfnamefont {R.~W.}\ \bibnamefont {Spekkens}},
  \ and\ \bibinfo {author} {\bibfnamefont {T.}~\bibnamefont {Fritz}},\
  }\href@noop {} {\bibfield  {journal} {\bibinfo  {journal} {Journal of Causal
  Inference}\ }\textbf {\bibinfo {volume} {7}},\ \bibinfo {pages} {20170020}
  (\bibinfo {year} {2019})}\BibitemShut {NoStop}%
\bibitem [{\citenamefont {Hall}(2011)}]{hall2011relaxed}%
  \BibitemOpen
  \bibfield  {author} {\bibinfo {author} {\bibfnamefont {M.~J.~W.}\
  \bibnamefont {Hall}},\ }\href@noop {} {\bibfield  {journal} {\bibinfo
  {journal} {Physical Review A}\ }\textbf {\bibinfo {volume} {84}},\ \bibinfo
  {pages} {022102} (\bibinfo {year} {2011})}\BibitemShut {NoStop}%
\bibitem [{\citenamefont {Gallicchio}\ \emph {et~al.}(2014)\citenamefont
  {Gallicchio}, \citenamefont {Friedman},\ and\ \citenamefont
  {Kaiser}}]{gallicchio2014testing}%
  \BibitemOpen
  \bibfield  {author} {\bibinfo {author} {\bibfnamefont {J.}~\bibnamefont
  {Gallicchio}}, \bibinfo {author} {\bibfnamefont {A.~S.}\ \bibnamefont
  {Friedman}}, \ and\ \bibinfo {author} {\bibfnamefont {D.~I.}\ \bibnamefont
  {Kaiser}},\ }\href@noop {} {\bibfield  {journal} {\bibinfo  {journal}
  {Physical Review Letters}\ }\textbf {\bibinfo {volume} {112}},\ \bibinfo
  {pages} {110405} (\bibinfo {year} {2014})}\BibitemShut {NoStop}%
\bibitem [{\citenamefont {P{\"u}tz}\ \emph {et~al.}(2014)\citenamefont
  {P{\"u}tz}, \citenamefont {Rosset}, \citenamefont {Barnea}, \citenamefont
  {Liang},\ and\ \citenamefont {Gisin}}]{putz2014arbitrarily}%
  \BibitemOpen
  \bibfield  {author} {\bibinfo {author} {\bibfnamefont {G.}~\bibnamefont
  {P{\"u}tz}}, \bibinfo {author} {\bibfnamefont {D.}~\bibnamefont {Rosset}},
  \bibinfo {author} {\bibfnamefont {T.~J.}\ \bibnamefont {Barnea}}, \bibinfo
  {author} {\bibfnamefont {Y.-C.}\ \bibnamefont {Liang}}, \ and\ \bibinfo
  {author} {\bibfnamefont {N.}~\bibnamefont {Gisin}},\ }\href@noop {}
  {\bibfield  {journal} {\bibinfo  {journal} {Physical Review Letters}\
  }\textbf {\bibinfo {volume} {113}},\ \bibinfo {pages} {190402} (\bibinfo
  {year} {2014})}\BibitemShut {NoStop}%
\bibitem [{\citenamefont {Friedman}\ \emph {et~al.}(2019)\citenamefont
  {Friedman}, \citenamefont {Guth}, \citenamefont {Hall}, \citenamefont
  {Kaiser},\ and\ \citenamefont {Gallicchio}}]{friedman2019relaxed}%
  \BibitemOpen
  \bibfield  {author} {\bibinfo {author} {\bibfnamefont {A.~S.}\ \bibnamefont
  {Friedman}}, \bibinfo {author} {\bibfnamefont {A.~H.}\ \bibnamefont {Guth}},
  \bibinfo {author} {\bibfnamefont {M.~J.~W.}\ \bibnamefont {Hall}}, \bibinfo
  {author} {\bibfnamefont {D.~I.}\ \bibnamefont {Kaiser}}, \ and\ \bibinfo
  {author} {\bibfnamefont {J.}~\bibnamefont {Gallicchio}},\ }\href@noop {}
  {\bibfield  {journal} {\bibinfo  {journal} {Physical Review A}\ }\textbf
  {\bibinfo {volume} {99}},\ \bibinfo {pages} {012121} (\bibinfo {year}
  {2019})}\BibitemShut {NoStop}%
\bibitem [{\citenamefont {Hall}(2010)}]{hall2010local}%
  \BibitemOpen
  \bibfield  {author} {\bibinfo {author} {\bibfnamefont {M.~J.~W.}\
  \bibnamefont {Hall}},\ }\href@noop {} {\bibfield  {journal} {\bibinfo
  {journal} {Physical Review Letters}\ }\textbf {\bibinfo {volume} {105}},\
  \bibinfo {pages} {250404} (\bibinfo {year} {2010})}\BibitemShut {NoStop}%
\bibitem [{\citenamefont {Handsteiner}\ \emph {et~al.}(2017)\citenamefont
  {Handsteiner}, \citenamefont {Friedman}, \citenamefont {Rauch}, \citenamefont
  {Gallicchio}, \citenamefont {Liu}, \citenamefont {Hosp}, \citenamefont
  {Kofler}, \citenamefont {Bricher}, \citenamefont {Fink}, \citenamefont
  {Leung} \emph {et~al.}}]{handsteiner2017cosmic}%
  \BibitemOpen
  \bibfield  {author} {\bibinfo {author} {\bibfnamefont {J.}~\bibnamefont
  {Handsteiner}}, \bibinfo {author} {\bibfnamefont {A.~S.}\ \bibnamefont
  {Friedman}}, \bibinfo {author} {\bibfnamefont {D.}~\bibnamefont {Rauch}},
  \bibinfo {author} {\bibfnamefont {J.}~\bibnamefont {Gallicchio}}, \bibinfo
  {author} {\bibfnamefont {B.}~\bibnamefont {Liu}}, \bibinfo {author}
  {\bibfnamefont {H.}~\bibnamefont {Hosp}}, \bibinfo {author} {\bibfnamefont
  {J.}~\bibnamefont {Kofler}}, \bibinfo {author} {\bibfnamefont
  {D.}~\bibnamefont {Bricher}}, \bibinfo {author} {\bibfnamefont
  {M.}~\bibnamefont {Fink}}, \bibinfo {author} {\bibfnamefont {C.}~\bibnamefont
  {Leung}},  \emph {et~al.},\ }\href@noop {} {\bibfield  {journal} {\bibinfo
  {journal} {Physical Review Letters}\ }\textbf {\bibinfo {volume} {118}},\
  \bibinfo {pages} {060401} (\bibinfo {year} {2017})}\BibitemShut {NoStop}%
\bibitem [{\citenamefont {Rauch}\ \emph {et~al.}(2018)\citenamefont {Rauch},
  \citenamefont {Handsteiner}, \citenamefont {Hochrainer}, \citenamefont
  {Gallicchio}, \citenamefont {Friedman}, \citenamefont {Leung}, \citenamefont
  {Liu}, \citenamefont {Bulla}, \citenamefont {Ecker}, \citenamefont
  {Steinlechner} \emph {et~al.}}]{rauch2018cosmic}%
  \BibitemOpen
  \bibfield  {author} {\bibinfo {author} {\bibfnamefont {D.}~\bibnamefont
  {Rauch}}, \bibinfo {author} {\bibfnamefont {J.}~\bibnamefont {Handsteiner}},
  \bibinfo {author} {\bibfnamefont {A.}~\bibnamefont {Hochrainer}}, \bibinfo
  {author} {\bibfnamefont {J.}~\bibnamefont {Gallicchio}}, \bibinfo {author}
  {\bibfnamefont {A.~S.}\ \bibnamefont {Friedman}}, \bibinfo {author}
  {\bibfnamefont {C.}~\bibnamefont {Leung}}, \bibinfo {author} {\bibfnamefont
  {B.}~\bibnamefont {Liu}}, \bibinfo {author} {\bibfnamefont {L.}~\bibnamefont
  {Bulla}}, \bibinfo {author} {\bibfnamefont {S.}~\bibnamefont {Ecker}},
  \bibinfo {author} {\bibfnamefont {F.}~\bibnamefont {Steinlechner}},  \emph
  {et~al.},\ }\href@noop {} {\bibfield  {journal} {\bibinfo  {journal}
  {Physical Review Letters}\ }\textbf {\bibinfo {volume} {121}},\ \bibinfo
  {pages} {080403} (\bibinfo {year} {2018})}\BibitemShut {NoStop}%
\bibitem [{\citenamefont {Abell{\'a}n}\ \emph {et~al.}(2018)\citenamefont
  {Abell{\'a}n}, \citenamefont {Acin}, \citenamefont {Alarc{\'o}n},
  \citenamefont {Alibart}, \citenamefont {Andersen}, \citenamefont {Andreoli},
  \citenamefont {Beckert}, \citenamefont {Beduini}, \citenamefont {Bendersky},
  \citenamefont {Bentivegna} \emph {et~al.}}]{abellan2018challenging}%
  \BibitemOpen
  \bibfield  {author} {\bibinfo {author} {\bibfnamefont {C.}~\bibnamefont
  {Abell{\'a}n}}, \bibinfo {author} {\bibfnamefont {A.}~\bibnamefont {Acin}},
  \bibinfo {author} {\bibfnamefont {A.}~\bibnamefont {Alarc{\'o}n}}, \bibinfo
  {author} {\bibfnamefont {O.}~\bibnamefont {Alibart}}, \bibinfo {author}
  {\bibfnamefont {C.}~\bibnamefont {Andersen}}, \bibinfo {author}
  {\bibfnamefont {F.}~\bibnamefont {Andreoli}}, \bibinfo {author}
  {\bibfnamefont {A.}~\bibnamefont {Beckert}}, \bibinfo {author} {\bibfnamefont
  {F.}~\bibnamefont {Beduini}}, \bibinfo {author} {\bibfnamefont
  {A.}~\bibnamefont {Bendersky}}, \bibinfo {author} {\bibfnamefont
  {M.}~\bibnamefont {Bentivegna}},  \emph {et~al.},\ }\href@noop {} {\bibfield
  {journal} {\bibinfo  {journal} {arXiv preprint arXiv:1805.04431}\ } (\bibinfo
  {year} {2018})}\BibitemShut {NoStop}%
\bibitem [{\citenamefont {Fritz}(2012)}]{fritz2012beyond}%
  \BibitemOpen
  \bibfield  {author} {\bibinfo {author} {\bibfnamefont {T.}~\bibnamefont
  {Fritz}},\ }\href@noop {} {\bibfield  {journal} {\bibinfo  {journal} {New
  Journal of Physics}\ }\textbf {\bibinfo {volume} {14}},\ \bibinfo {pages}
  {103001} (\bibinfo {year} {2012})}\BibitemShut {NoStop}%
\bibitem [{\citenamefont {Chaves}(2016)}]{chaves2016polynomial}%
  \BibitemOpen
  \bibfield  {author} {\bibinfo {author} {\bibfnamefont {R.}~\bibnamefont
  {Chaves}},\ }\href@noop {} {\bibfield  {journal} {\bibinfo  {journal}
  {Physical Review Letters}\ }\textbf {\bibinfo {volume} {116}},\ \bibinfo
  {pages} {010402} (\bibinfo {year} {2016})}\BibitemShut {NoStop}%
\bibitem [{\citenamefont {Rosset}\ \emph {et~al.}(2016)\citenamefont {Rosset},
  \citenamefont {Branciard}, \citenamefont {Barnea}, \citenamefont {P{\"u}tz},
  \citenamefont {Brunner},\ and\ \citenamefont {Gisin}}]{rosset2016nonlinear}%
  \BibitemOpen
  \bibfield  {author} {\bibinfo {author} {\bibfnamefont {D.}~\bibnamefont
  {Rosset}}, \bibinfo {author} {\bibfnamefont {C.}~\bibnamefont {Branciard}},
  \bibinfo {author} {\bibfnamefont {T.~J.}\ \bibnamefont {Barnea}}, \bibinfo
  {author} {\bibfnamefont {G.}~\bibnamefont {P{\"u}tz}}, \bibinfo {author}
  {\bibfnamefont {N.}~\bibnamefont {Brunner}}, \ and\ \bibinfo {author}
  {\bibfnamefont {N.}~\bibnamefont {Gisin}},\ }\href@noop {} {\bibfield
  {journal} {\bibinfo  {journal} {Physical Review Letters}\ }\textbf {\bibinfo
  {volume} {116}},\ \bibinfo {pages} {010403} (\bibinfo {year}
  {2016})}\BibitemShut {NoStop}%
\bibitem [{\citenamefont {Renou}\ \emph {et~al.}(2019)\citenamefont {Renou},
  \citenamefont {B{\"a}umer}, \citenamefont {Boreiri}, \citenamefont {Brunner},
  \citenamefont {Gisin},\ and\ \citenamefont {Beigi}}]{renou2019genuine}%
  \BibitemOpen
  \bibfield  {author} {\bibinfo {author} {\bibfnamefont {M.-O.}\ \bibnamefont
  {Renou}}, \bibinfo {author} {\bibfnamefont {E.}~\bibnamefont {B{\"a}umer}},
  \bibinfo {author} {\bibfnamefont {S.}~\bibnamefont {Boreiri}}, \bibinfo
  {author} {\bibfnamefont {N.}~\bibnamefont {Brunner}}, \bibinfo {author}
  {\bibfnamefont {N.}~\bibnamefont {Gisin}}, \ and\ \bibinfo {author}
  {\bibfnamefont {S.}~\bibnamefont {Beigi}},\ }\href@noop {} {\bibfield
  {journal} {\bibinfo  {journal} {Physical Review Letters}\ }\textbf {\bibinfo
  {volume} {123}},\ \bibinfo {pages} {140401} (\bibinfo {year}
  {2019})}\BibitemShut {NoStop}%
\bibitem [{\citenamefont {Barrett}\ \emph {et~al.}(2021)\citenamefont
  {Barrett}, \citenamefont {Lorenz},\ and\ \citenamefont
  {Oreshkov}}]{barrett2021cyclic}%
  \BibitemOpen
  \bibfield  {author} {\bibinfo {author} {\bibfnamefont {J.}~\bibnamefont
  {Barrett}}, \bibinfo {author} {\bibfnamefont {R.}~\bibnamefont {Lorenz}}, \
  and\ \bibinfo {author} {\bibfnamefont {O.}~\bibnamefont {Oreshkov}},\
  }\href@noop {} {\bibfield  {journal} {\bibinfo  {journal} {Nature
  Communications}\ }\textbf {\bibinfo {volume} {12}},\ \bibinfo {pages} {1}
  (\bibinfo {year} {2021})}\BibitemShut {NoStop}%
\bibitem [{\citenamefont {Gachechiladze}\ \emph {et~al.}(2020)\citenamefont
  {Gachechiladze}, \citenamefont {Miklin},\ and\ \citenamefont
  {Chaves}}]{PhysRevLett.125.230401}%
  \BibitemOpen
  \bibfield  {author} {\bibinfo {author} {\bibfnamefont {M.}~\bibnamefont
  {Gachechiladze}}, \bibinfo {author} {\bibfnamefont {N.}~\bibnamefont
  {Miklin}}, \ and\ \bibinfo {author} {\bibfnamefont {R.}~\bibnamefont
  {Chaves}},\ }\href {\doibase 10.1103/PhysRevLett.125.230401} {\bibfield
  {journal} {\bibinfo  {journal} {Physical Review Letters}\ }\textbf {\bibinfo
  {volume} {125}},\ \bibinfo {pages} {230401} (\bibinfo {year}
  {2020})}\BibitemShut {NoStop}%
\bibitem [{\citenamefont {Henson}\ \emph {et~al.}(2014)\citenamefont {Henson},
  \citenamefont {Lal},\ and\ \citenamefont {Pusey}}]{Henson_2014}%
  \BibitemOpen
  \bibfield  {author} {\bibinfo {author} {\bibfnamefont {J.}~\bibnamefont
  {Henson}}, \bibinfo {author} {\bibfnamefont {R.}~\bibnamefont {Lal}}, \ and\
  \bibinfo {author} {\bibfnamefont {M.~F.}\ \bibnamefont {Pusey}},\ }\href
  {\doibase 10.1088/1367-2630/16/11/113043} {\bibfield  {journal} {\bibinfo
  {journal} {New Journal of Physics}\ }\textbf {\bibinfo {volume} {16}},\
  \bibinfo {pages} {113043} (\bibinfo {year} {2014})}\BibitemShut {NoStop}%
\bibitem [{\citenamefont {Eberhard}(1993)}]{eberhard1993background}%
  \BibitemOpen
  \bibfield  {author} {\bibinfo {author} {\bibfnamefont {P.~H.}\ \bibnamefont
  {Eberhard}},\ }\href@noop {} {\bibfield  {journal} {\bibinfo  {journal}
  {Physical Review A}\ }\textbf {\bibinfo {volume} {47}},\ \bibinfo {pages}
  {R747} (\bibinfo {year} {1993})}\BibitemShut {NoStop}%
\bibitem [{\citenamefont {Liu}\ \emph {et~al.}(2018)\citenamefont {Liu},
  \citenamefont {Zhao}, \citenamefont {Li}, \citenamefont {Guan}, \citenamefont
  {Zhang}, \citenamefont {Bai}, \citenamefont {Zhang}, \citenamefont {Liu},
  \citenamefont {Wu}, \citenamefont {Yuan} \emph {et~al.}}]{liu2018device}%
  \BibitemOpen
  \bibfield  {author} {\bibinfo {author} {\bibfnamefont {Y.}~\bibnamefont
  {Liu}}, \bibinfo {author} {\bibfnamefont {Q.}~\bibnamefont {Zhao}}, \bibinfo
  {author} {\bibfnamefont {M.-H.}\ \bibnamefont {Li}}, \bibinfo {author}
  {\bibfnamefont {J.-Y.}\ \bibnamefont {Guan}}, \bibinfo {author}
  {\bibfnamefont {Y.}~\bibnamefont {Zhang}}, \bibinfo {author} {\bibfnamefont
  {B.}~\bibnamefont {Bai}}, \bibinfo {author} {\bibfnamefont {W.}~\bibnamefont
  {Zhang}}, \bibinfo {author} {\bibfnamefont {W.-Z.}\ \bibnamefont {Liu}},
  \bibinfo {author} {\bibfnamefont {C.}~\bibnamefont {Wu}}, \bibinfo {author}
  {\bibfnamefont {X.}~\bibnamefont {Yuan}},  \emph {et~al.},\ }\href@noop {}
  {\bibfield  {journal} {\bibinfo  {journal} {Nature}\ }\textbf {\bibinfo
  {volume} {562}},\ \bibinfo {pages} {548} (\bibinfo {year}
  {2018})}\BibitemShut {NoStop}%
\bibitem [{\citenamefont {Massar}(2002)}]{PhysRevA.65.032121}%
  \BibitemOpen
  \bibfield  {author} {\bibinfo {author} {\bibfnamefont {S.}~\bibnamefont
  {Massar}},\ }\href {\doibase 10.1103/PhysRevA.65.032121} {\bibfield
  {journal} {\bibinfo  {journal} {Physical Review A}\ }\textbf {\bibinfo
  {volume} {65}},\ \bibinfo {pages} {032121} (\bibinfo {year}
  {2002})}\BibitemShut {NoStop}%
\bibitem [{\citenamefont {P\'al}\ and\ \citenamefont
  {V\'ertesi}(2015)}]{PhysRevA.92.052104}%
  \BibitemOpen
  \bibfield  {author} {\bibinfo {author} {\bibfnamefont {K.~F.}\ \bibnamefont
  {P\'al}}\ and\ \bibinfo {author} {\bibfnamefont {T.}~\bibnamefont
  {V\'ertesi}},\ }\href {\doibase 10.1103/PhysRevA.92.052104} {\bibfield
  {journal} {\bibinfo  {journal} {Physical Review A}\ }\textbf {\bibinfo
  {volume} {92}},\ \bibinfo {pages} {052104} (\bibinfo {year}
  {2015})}\BibitemShut {NoStop}%
\bibitem [{\citenamefont {Cao}\ \emph {et~al.}(2016)\citenamefont {Cao},
  \citenamefont {Zhou}, \citenamefont {Yuan},\ and\ \citenamefont
  {Ma}}]{cao2016source}%
  \BibitemOpen
  \bibfield  {author} {\bibinfo {author} {\bibfnamefont {Z.}~\bibnamefont
  {Cao}}, \bibinfo {author} {\bibfnamefont {H.}~\bibnamefont {Zhou}}, \bibinfo
  {author} {\bibfnamefont {X.}~\bibnamefont {Yuan}}, \ and\ \bibinfo {author}
  {\bibfnamefont {X.}~\bibnamefont {Ma}},\ }\href@noop {} {\bibfield  {journal}
  {\bibinfo  {journal} {Physical Review X}\ }\textbf {\bibinfo {volume} {6}},\
  \bibinfo {pages} {011020} (\bibinfo {year} {2016})}\BibitemShut {NoStop}%
\bibitem [{\citenamefont {Balke}\ and\ \citenamefont
  {Pearl}(1997)}]{balke1997bounds}%
  \BibitemOpen
  \bibfield  {author} {\bibinfo {author} {\bibfnamefont {A.}~\bibnamefont
  {Balke}}\ and\ \bibinfo {author} {\bibfnamefont {J.}~\bibnamefont {Pearl}},\
  }\href@noop {} {\bibfield  {journal} {\bibinfo  {journal} {Journal of the
  American Statistical Association}\ }\textbf {\bibinfo {volume} {92}},\
  \bibinfo {pages} {1171} (\bibinfo {year} {1997})}\BibitemShut {NoStop}%
\bibitem [{\citenamefont {Clauser}\ \emph {et~al.}(1969)\citenamefont
  {Clauser}, \citenamefont {Horne}, \citenamefont {Shimony},\ and\
  \citenamefont {Holt}}]{clauser1969proposed}%
  \BibitemOpen
  \bibfield  {author} {\bibinfo {author} {\bibfnamefont {J.~F.}\ \bibnamefont
  {Clauser}}, \bibinfo {author} {\bibfnamefont {M.~A.}\ \bibnamefont {Horne}},
  \bibinfo {author} {\bibfnamefont {A.}~\bibnamefont {Shimony}}, \ and\
  \bibinfo {author} {\bibfnamefont {R.~A.}\ \bibnamefont {Holt}},\ }\href@noop
  {} {\bibfield  {journal} {\bibinfo  {journal} {Physical review letters}\
  }\textbf {\bibinfo {volume} {23}},\ \bibinfo {pages} {880} (\bibinfo {year}
  {1969})}\BibitemShut {NoStop}%
\bibitem [{\citenamefont {Cao}\ and\ \citenamefont
  {Peng}(2016)}]{PhysRevA.94.042126}%
  \BibitemOpen
  \bibfield  {author} {\bibinfo {author} {\bibfnamefont {Z.}~\bibnamefont
  {Cao}}\ and\ \bibinfo {author} {\bibfnamefont {T.}~\bibnamefont {Peng}},\
  }\href {\doibase 10.1103/PhysRevA.94.042126} {\bibfield  {journal} {\bibinfo
  {journal} {Physical Review A}\ }\textbf {\bibinfo {volume} {94}},\ \bibinfo
  {pages} {042126} (\bibinfo {year} {2016})}\BibitemShut {NoStop}%
\bibitem [{\citenamefont {Preskill}(1992)}]{Preskill1992Do}%
  \BibitemOpen
  \bibfield  {author} {\bibinfo {author} {\bibfnamefont {J.}~\bibnamefont
  {Preskill}},\ }\href@noop {} {\bibfield  {journal} {\bibinfo  {journal} {An
  International Symposium on Black Holes, Membranes, Wormholes, and
  Superstrings}\ } (\bibinfo {year} {1992})}\BibitemShut {NoStop}%
\end{thebibliography}%


\renewcommand{\theequation}{S\arabic{equation}}
\setcounter{equation}{0}
\renewcommand{\thefigure}{S\arabic{figure}}
\setcounter{figure}{0}
\renewcommand{\thesection}{S\arabic{section}}
\setcounter{section}{0}

\newpage

\onecolumngrid


\vspace*{2cm}
\begin{center}
	{\bf \Large 
		Supplemental Material to\vspace*{0.3cm}\\ 
		\emph{Detection loophole in quantum causality and its countermeasures}
	}
\end{center}

\section{Proof of Theorem 1}
\label{AppSec:1}
The proof consists of two steps. In the first step, we show that, for the case $a,b,x\in \{0, 1\}$, no classical correlation can violate the inequality 
Eq.~(8) in the main text even if the detection efficiency is imperfect and the detectors are untrusted.
In the second step, we show that as long as $\eta>95.97\%$, there exists a 
quantum correlation violating this inequality even with trusted detectors. 

Let us start with the first step. We will prove
\begin{equation}
\label{eq:CACE1}
p(1| do(1)) -  p(1| do(0)) \ge 2 p(0,0 | 0 ) + p(1,1 |0) + p(0,1|1) + p( 1,1|1) -2
\end{equation}
holds classically even if the detection efficiency is imperfect.

 By the linearity of Eq.~\eqref{eq:CACE1}, we only need to consider deterministic classical strategies. 
This implies two facts. First, the right hand side (RHS) of Eq.~\eqref{eq:CACE1} can be rewritten as
\begin{equation*}
2  p_A(0|0) p_B(0|0) + p_A(1|0)  p_B(1|1)  +  p_A(0|1) p_B(1|0) +  p_A( 1| 1) p_B(1|1) - 2
\end{equation*}
by utilizing Eq.~(2) in the main text.
Secondly, we can use two deterministic functions $F_A, F_B : \{0,1\} \to \{0, 1, \Phi\}$ to determine the probability distribution of $A$ and $B$, where $\Phi$ denotes an empty outcome. 
Namely, $ p_C( c|z) = 1$ if and only $F_C(z)=c$, where $C=A,B$. 
We now divide into nine cases: (1) $F_B(1)=1$ and $F_B(0)=0$; (2) $F_B(1)=0$ and $F_B(0) = 0$;  (3) $F_B(1) = 0$ and $F_B(0)=1$; (4) $F_B(1) =1$ and $F_B(0)=1$; (5) $F_B(1) = \Phi$ and $F_B(0)= \Phi$; (6) $F_B(1) = \Phi$ and $F_B(0)=0$; (7) $F_B(1) = \Phi$ and $F_B(0)=1$; (8) $F_B(1) =0$ and $F_B(0)= \Phi$; (9) $F_B(1) =1$ and $F_B(0)= \Phi$. 

{\it Case 1.}  $F_B(1)=1$ and $F_B(0)=0$. In this case, the RHS of Eq.~\eqref{eq:CACE1} is reduced to
\begin{equation}
2 p_A(0|0) +  p_A(1|0)  + p_A( 1| 1) - 2,
\end{equation}
and the left hand side (LHS) of Eq.~\eqref{eq:CACE1} is reduced to 1.
Clearly, $F_A(0)=0$ and $F_A(1)=1$ maximize the RHS of Eq.~\eqref{eq:CACE1}, which becomes
\begin{equation}
2 p_A(0|0) +  p_A(1|0)  + p_A( 1| 1) - 2 = 2+0+1-2 = 1.
\end{equation}
As the maximum of the RHS of Eq.~\eqref{eq:CACE1}  is no larger than the LHS of Eq.~\eqref{eq:CACE1}, Eq.~\eqref{eq:CACE1} always holds in Case 1.

{\it Case 2.}  $F_B(1)=0$ and $F_B(0) = 0$.  In this case, the RHS of Eq.~\eqref{eq:CACE1} is reduced to
\begin{equation}
2 p_A(0|0)  - 2 \le 2-2=0,
\end{equation}
and the LHS of Eq.~\eqref{eq:CACE1} is reduced to 0.
As the maximum of the RHS of Eq.~\eqref{eq:CACE1}  is no larger than the LHS of Eq.~\eqref{eq:CACE1}, Eq.~\eqref{eq:CACE1} always holds in Case 2.

{\it Case 3.}  $F_B(1)=0$ and $F_B(0) = 1$.  In this case, the RHS of Eq.~\eqref{eq:CACE1} is reduced to
\begin{equation}
 p_A(0|1)  - 2 \le 1-2 = -1,
\end{equation}
and the LHS of Eq.~\eqref{eq:CACE1} is reduced to $-1$.
As the maximum of the RHS of Eq.~\eqref{eq:CACE1}  is no larger than the LHS of Eq.~\eqref{eq:CACE1}, Eq.~\eqref{eq:CACE1} always holds in Case 3.

{\it Case 4.}  $F_B(1)=1$ and $F_B(0) = 1$.  In this case, the RHS of Eq.~\eqref{eq:CACE1} is reduced to
\begin{equation}
p_A(1|0) +  p_A(0|1) +  p_A( 1| 1) - 2,
\end{equation}
and the LHS of Eq.~\eqref{eq:CACE1} is reduced to 0.
Clearly, $F_A(0)=1$ and $F_A(1)=0$ maximize the RHS of Eq.~\eqref{eq:CACE1}, which becomes
\begin{equation}
p_A(1|0) +  p_A(0|1) +  p_A( 1| 1) - 2=1+1+0-2= 0.
\end{equation}
As the maximum of the RHS of Eq.~\eqref{eq:CACE1}  is no larger than the LHS of Eq.~\eqref{eq:CACE1}, Eq.~\eqref{eq:CACE1} always holds in Case 4.
 
{\it Case 5.}   $F_B(1) = \Phi$ and $F_B(0)= \Phi$. In this case, the RHS of Eq.~\eqref{eq:CACE1} is reduced to
$-2$,
and the LHS of Eq.~\eqref{eq:CACE1} is reduced to 0.
As the RHS of Eq.~\eqref{eq:CACE1}  is no larger than the LHS of Eq.~\eqref{eq:CACE1}, Eq.~\eqref{eq:CACE1} always holds in Case 5.

{\it Case 6.}  $F_B(1) = \Phi$ and $F_B(0)=0$. 
In this case, the RHS of Eq.~\eqref{eq:CACE1} is reduced to
\begin{equation}
2p(0,0|0)-2,
\end{equation}
and the LHS of Eq.~\eqref{eq:CACE1} is reduced to 0.
Clearly, $F_A(0)=0$ maximize the RHS of Eq.~\eqref{eq:CACE1}, which becomes
\begin{equation}
2p(0,0|0)-2=2-2= 0.
\end{equation}
As the maximum of the RHS of Eq.~\eqref{eq:CACE1}  is no larger than the LHS of Eq.~\eqref{eq:CACE1}, Eq.~\eqref{eq:CACE1} always holds in Case 6.

{\it Case 7.}  $F_B(1) = \Phi$ and $F_B(0)=1$. 
In this case, the RHS of Eq.~\eqref{eq:CACE1} is reduced to
\begin{equation}
p(0,1|1)-2,
\end{equation}
and the LHS of Eq.~\eqref{eq:CACE1} is reduced to $-1$.
Clearly, $F_A(1)=0$ maximize the RHS of Eq.~\eqref{eq:CACE1}, which becomes
\begin{equation}
p(0,1|1)-2=1-2= -1.
\end{equation}
As the maximum of the RHS of Eq.~\eqref{eq:CACE1}  is no larger than the LHS of Eq.~\eqref{eq:CACE1}, Eq.~\eqref{eq:CACE1} always holds in Case 7.

{\it Case 8.}  $F_B(1) =0$ and $F_B(0)= \Phi$. In this case, the RHS of Eq.~\eqref{eq:CACE1} is reduced to
$-2$,
and the LHS of Eq.~\eqref{eq:CACE1} is reduced to 0.
As the RHS of Eq.~\eqref{eq:CACE1}  is no larger than the LHS of Eq.~\eqref{eq:CACE1}, Eq.~\eqref{eq:CACE1} always holds in Case 8.

{\it Case 9.}  $F_B(1) =1$ and $F_B(0)= \Phi$. 
In this case, the RHS of Eq.~\eqref{eq:CACE1} is reduced to
\begin{equation}
p(1,1|0)+p(1,1|1)-2,
\end{equation}
and the LHS of Eq.~\eqref{eq:CACE1} is reduced to $1$.
Clearly, $F_A(1)=1$ and $F_A(0)=1$ maximize the RHS of Eq.~\eqref{eq:CACE1}, which becomes
\begin{equation}
p(1,1|0)+p(1,1|1)-2=1+1-2= 0.
\end{equation}
As the maximum of the RHS of Eq.~\eqref{eq:CACE1}  is no larger than the LHS of Eq.~\eqref{eq:CACE1}, Eq.~\eqref{eq:CACE1} always holds in Case 9.

Now let us move on to the second step. For the quantum scenario, let  CACE$^*$ denote the RHS of Eq.~(8) in the main text and QACE denote the LHS of Eq.~(8) in the main text.
Let us recall the setting in which the violation $\textrm{CACE}^* - \textrm{QACE}$ obtains the largest value $3-2\sqrt{2}$ \cite{PhysRevLett.125.230401}. According to Eq.~(4) in the main text, we only need to specify $\rho_{AB}$, $M^x_a$ and $N^a_b$.
The quantum state $\rho_{AB}$ is set as $\ket{\psi} = \cos(\alpha)  \ket{0, 0} + \sin(\alpha) \ket{1,1}$. 
Let the measurement settings on $A$ be
\begin{eqnarray}
M_0^x =  \frac{\mathbbm{1}}{2} +  \sin(\theta_x) \frac{\sigma_x}{2} + \cos(\theta_x) \frac{\sigma_z}{2} ,  \\
M_1^x  =   \frac{\mathbbm{1}}{2}- \sin(\theta_x) \frac{\sigma_x}{2} -  \cos(\theta_x) \frac{\sigma_z}{2}, \nonumber
\end{eqnarray}
where $x\in \{0,1\}$, and the measurement settings on $B$ be
\begin{eqnarray}
N_0^a =  \frac{\mathbbm{1}}{2} +   \sin(\phi_a) \frac{\sigma_x}{2} +  \cos(\phi_a) \frac{\sigma_z}{2} ,  \\
N_1^a  =  \frac{\mathbbm{1}}{2}-  \sin(\phi_a) \frac{\sigma_x}{2} -  \cos(\phi_a) \frac{\sigma_z}{2}, \nonumber
\end{eqnarray}
where $a\in \{0,1\}$.
Let  $\langle O \rangle$ be a shorthand for $\bra{\psi} O \ket{\psi}$. Then QACE can be expressed as $| \langle \mathbbm{1} \otimes (N^0_0 - N^1_0) \rangle  | $
and CACE$^*$ can be expressed as
\begin{equation}
2 \langle M_0^0 \otimes N_0^0 \rangle + \langle M_1^0 \otimes N_1^1 \rangle + \langle M_0^1 \otimes N_1^0 \rangle + \langle M_1 \otimes N_1^1 \rangle -2.
\end{equation}

Let $\phi_1 = -\phi_0$ and $\theta_1 = -\pi /2$. Then QACE is simplified to 0 and  CACE$^*$ is simplified to
\begin{eqnarray}
\frac{1}{4}\{5+ (\cos \theta_0 - \cos \phi_0) \cos(2\alpha) +3\cos(\theta_0)\cos(\phi_0)  +\sin(2\alpha)\sin(\phi_0)[2+\sin(\theta_0)]\} -2.
\end{eqnarray}
When $\alpha= 0.7165, \phi_0= 0.6750, \theta_0=0.2447$, we obtain 
$\textrm{CACE}^* - \textrm{QACE} = 3-2\sqrt{2} = 0.1716$.

In the quantum setting with detector efficiency $\eta$, we have
\begin{eqnarray}
p_\eta (a,b |x) &=& \eta^2 p (a,b|x),   \\
p_\eta (b | do(a) )  &=&  \eta p (b | do(a) ) , \nonumber
\end{eqnarray}
where $p$ is the probability with perfect quantum detection, and $p_\eta$ is the probability with  
detection efficiency $\eta$. 
With imperfect detection efficiency $\eta$, let CACE$^*_\eta$ denote the RHS of Eq.~(8) in the main text and QACE$_\eta$ denote the LHS of Eq.~(8) in the main text.
Then we have 
\begin{eqnarray}
 \textrm{QACE}_\eta &=&  \max\limits_{b,a,a'}(p_\eta(b| do(a))-p_\eta(b| do(a'))  \nonumber \\
  &=& \eta \max\limits_{b,a,a'}(p(b| do(a))-p(b| do(a'))  \\
&=&  \eta \times  \textrm{QACE},  \nonumber
\end{eqnarray}
  and 
\begin{eqnarray}
\textrm{CACE}^*_\eta  & =& 2 p_\eta(0,0 | 0 ) + p_\eta(1,1 |0) + p_\eta(0,1|1) + p_\eta( 1,1|1) -2 \nonumber \\
& =& \eta^2[ 2 p(0,0 | 0 ) + p(1,1 |0) + p(0,1|1) + p( 1,1|1)] -2  \\
& =& \eta^2(2 \langle M_0^0 \otimes N_0^0 \rangle + \langle M_1^0 \otimes N_1^1 \rangle + \langle M_0^1 \otimes N_1^0 \rangle + \langle M_1^1 \otimes N_1^1 \rangle) -2. \nonumber 
\end{eqnarray}
The violation condition $\textrm{CACE}_\eta^* - \textrm{QACE}_\eta>0$ leads to $ 2.1716 \eta^2 > 2$. Since $\eta$ is positive, this shows $\eta  >  95.97\%$, which proves the theorem.

\section{Proof of Theorem 2}
In the nosignaling setting, let $\textrm{CACE}^*$ denote the RHS of Eq.~(8) in the main text and $\textrm{NACE}$ the LHS of Eq.~(8) in the main text.
Similar to the quantum case, in the imperfection detection setting with detection efficiency $\eta$, we have
\begin{eqnarray}
p^{N}_\eta (a,b |x) &=&\eta^2 p^{N} (a,b|x),   \\
p^{N}_\eta (b | do(a) ) &=&\eta p^{N} (b | do(a) ), \nonumber
\end{eqnarray}
where $p^N$ denotes a no-signaling probability distribution with perfect detection efficiency and $p^N_\eta$ denotes a no-signaling probability distribution with detection efficiency $\eta$.

Let us set $p(0,0 | 0 ) = p(1,1 |0) = p(0,1|1) = p( 1,1|1) = 1/2 $.  We first show that this assignment satisfies the nosignaling constraint.
By mapping this probability distribution to the Bell setting according to
\begin{equation}
p( a, b|x )= p_{Bell}(a ,b |x, a),
\end{equation}
we have $p_{Bell}(0,0 | 0,0 ) = p_{Bell}(1,1 |0,1) = p_{Bell}(0,1|1,0) = p_{Bell}( 1,1|1,1) = 1/2 $. 
In addition, we let $p_{Bell}(1,1 | 0,0 ) = p_{Bell}(0,0 |0,1) = p_{Bell}(1,0|1,0) = p_{Bell}( 0,0|1,1) = 1/2 $.
All other values of $p_{Bell}(a,b|x,y)$ are set to 0.
Clearly $p_{Bell}(a|x,y)=p_{Bell}(b|x,y)=1/2$ for all $x,y,a,b$, hence the no-signaling condition holds.
An illustration of this probabilitiy distribution is shown in Fig.~\ref{fig:nosignal222}.

\begin{figure}[htb]
\centering \includegraphics[width=5cm]{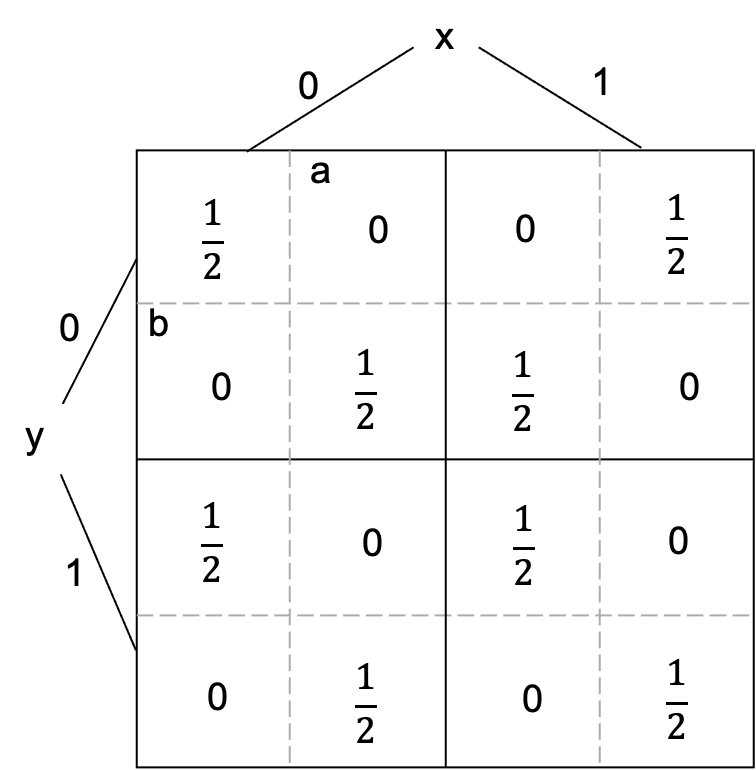}
\caption{Illustration of the no-signaling probability distribution which maximally violates the classical causal bound. } 
\label{fig:nosignal222}
\end{figure}

For this no-signaling correlation, we have 
\begin{equation}
\textrm{CACE}^* = 2\times \frac{1}{2} + \frac{1}{2} +  \frac{1}{2} +  \frac{1}{2} -2 =  \frac{1}{2}. 
\end{equation}
In addition, since $p(b|do(a))=1/2$ for all $a,b$, we have $\textrm{NACE}=0$. 
Therefore, this no-signaling correlation violates the classical bound by $\textrm{CACE}^*  - \textrm{NACE} = 0.5$. 

When detection efficiency is considered, we have $ \textrm{NACE}_\eta = \eta \times  \textrm{NACE} $, and 
\begin{equation*}
\textrm{CACE}^*_\eta  = \eta^2(2\times \frac{1}{2} + \frac{1}{2} +  \frac{1}{2} +  \frac{1}{2}) -2.
\end{equation*}
The violation condition $\textrm{CACE}_\eta^* - \textrm{NACE}_\eta$ leads to $ 2.5 \eta^2 \ge 2$. Since $\eta$ is positive, this shows $\eta  >  89.44\%$, completing the proof.

\section{Proof of Theorem 3}
We first show for classical correlations, we have 
\begin{equation}
\label{eq:CACE2}
p_B(1|1) - p_B(1|0) \ge I_{222}.
\end{equation}

Since the inequality to be proved is linear, we only need to consider deterministic strategies. We will exploit the relations
\begin{equation}
\sum\limits_{a,b=0}^1 p(a,b|x) = \eta^2, \quad \forall x.
\end{equation}
frequently.
We now divide into nine cases: (1) $F_B(1)=1$ and $F_B(0)=0$; (2) $F_B(1)=0$ and $F_B(0) = 0$;  (3) $F_B(1) = 0$ and $F_B(0)=1$; (4) $F_B(1) =1$ and $F_B(0)=1$; (5) $F_B(1) = \Phi$ and $F_B(0)= \Phi$; (6) $F_B(1) = \Phi$ and $F_B(0)=0$; (7) $F_B(1) = \Phi$ and $F_B(0)=1$; (8) $F_B(1) =0$ and $F_B(0)= \Phi$; (9) $F_B(1) =1$ and $F_B(0)= \Phi$. 

{\it Case 1.}  $F_B(1)=1$ and $F_B(0)=0$. In this case, the RHS of Eq.~\eqref{eq:CACE2} is reduced to
\begin{equation}
2 p(0,0 | 0 ) + p(1,1 |0)  + p( 1,1|1) - 1 -\eta^2,
\end{equation}
and the left hand side (LHS) of Eq.~\eqref{eq:CACE2} is reduced to 1.
Clearly, $F_A(0)=0$ and $F_A(1)=1$ maximize the RHS of Eq.~\eqref{eq:CACE2}, which becomes
\begin{equation}
2 p(0,0 | 0 ) + p( 1,1|1) - 1 -\eta^2 \le  p(0,0 | 0 ) \le 1.
\end{equation}
As the maximum of the RHS of Eq.~\eqref{eq:CACE2}  is no larger than the LHS of Eq.~\eqref{eq:CACE2}, Eq.~\eqref{eq:CACE2} always holds in Case 1.

{\it Case 2.}  $F_B(1)=0$ and $F_B(0) = 0$.  In this case, the RHS of Eq.~\eqref{eq:CACE2} is reduced to
\begin{equation}
2 p(0,0 | 0 ) - 1 -\eta^2,
\end{equation}
and the LHS of Eq.~\eqref{eq:CACE2} is reduced to 0.
As the maximum of the RHS of Eq.~\eqref{eq:CACE2}  is no larger than the LHS of Eq.~\eqref{eq:CACE2}, Eq.~\eqref{eq:CACE2} always holds in Case 2.

{\it Case 3.}  $F_B(1)=0$ and $F_B(0) = 1$.  In this case, the RHS of Eq.~\eqref{eq:CACE2} is reduced to
\begin{equation}
 p(0,1|1)  - 1 -\eta^2,
\end{equation}
and the LHS of Eq.~\eqref{eq:CACE2} is reduced to $-1$.
As the maximum of the RHS of Eq.~\eqref{eq:CACE2}  is no larger than the LHS of Eq.~\eqref{eq:CACE2}, Eq.~\eqref{eq:CACE2} always holds in Case 3.

{\it Case 4.}  $F_B(1)=1$ and $F_B(0) = 1$.  In this case, the RHS of Eq.~\eqref{eq:CACE2} is reduced to
\begin{eqnarray}
&& p(1,1 |0) + p(0,1|1) + p( 1,1|1) - 1 -\eta^2  \nonumber \\ 
&\le& p(1,1 |0) + [p(0,1|1) + p( 1,1|1)] - 1 -\eta^2  \\  
&\le& 1+\eta^2 -1-\eta^2=0, \nonumber
\end{eqnarray}
and the LHS of Eq.~\eqref{eq:CACE2} is reduced to 0.
As the maximum of the RHS of Eq.~\eqref{eq:CACE2}  is no larger than the LHS of Eq.~\eqref{eq:CACE2}, Eq.~\eqref{eq:CACE2} always holds in Case 4.
 
{\it Case 5.}   $F_B(1) = \Phi$ and $F_B(0)= \Phi$. In this case, the RHS of Eq.~\eqref{eq:CACE2} is reduced to
$- 1 -\eta^2$,
and the LHS of Eq.~\eqref{eq:CACE2} is reduced to 0.
As the RHS of Eq.~\eqref{eq:CACE2}  is no larger than the LHS of Eq.~\eqref{eq:CACE2}, Eq.~\eqref{eq:CACE2} always holds in Case 5.

{\it Case 6.}  $F_B(1) = \Phi$ and $F_B(0)=0$. 
In this case, the RHS of Eq.~\eqref{eq:CACE2} is reduced to
\begin{equation}
2 p(0,0 | 0 ) - 1 -\eta^2 \le 2\eta^2-1-\eta^2 \le 0
\end{equation}
and the LHS of Eq.~\eqref{eq:CACE2} is reduced to 0.
As the maximum of the RHS of Eq.~\eqref{eq:CACE2}  is no larger than the LHS of Eq.~\eqref{eq:CACE2}, Eq.~\eqref{eq:CACE2} always holds in Case 6.

{\it Case 7.}  $F_B(1) = \Phi$ and $F_B(0)=1$. 
In this case, the RHS of Eq.~\eqref{eq:CACE2} is reduced to
\begin{equation}
 p(0,1|1)- 1 -\eta^2 \le \eta^2-1-\eta^2 \le -1,
\end{equation}
and the LHS of Eq.~\eqref{eq:CACE2} is reduced to $-1$.
As the maximum of the RHS of Eq.~\eqref{eq:CACE2}  is no larger than the LHS of Eq.~\eqref{eq:CACE2}, Eq.~\eqref{eq:CACE2} always holds in Case 7.

{\it Case 8.}  $F_B(1) =0$ and $F_B(0)= \Phi$. In this case, the RHS of Eq.~\eqref{eq:CACE2} is reduced to
$- 1 -\eta^2$,
and the LHS of Eq.~\eqref{eq:CACE2} is reduced to 0.
As the RHS of Eq.~\eqref{eq:CACE2}  is no larger than the LHS of Eq.~\eqref{eq:CACE2}, Eq.~\eqref{eq:CACE2} always holds in Case 8.

{\it Case 9.}  $F_B(1) =1$ and $F_B(0)= \Phi$. 
In this case, the RHS of Eq.~\eqref{eq:CACE2} is reduced to
\begin{equation}
p(1,1 |0)  + p( 1,1|1) - 1 -\eta^2\le 2\eta^2-1-\eta^2 \le 0,
\end{equation}
and the LHS of Eq.~\eqref{eq:CACE2} is reduced to $1$.
As the maximum of the RHS of Eq.~\eqref{eq:CACE2}  is no larger than the LHS of Eq.~\eqref{eq:CACE2}, Eq.~\eqref{eq:CACE2} always holds in Case 9.

We use the same nosignaling distribution as the previous section. Namely, for perfect detection efficiency, we set $p(0,0 | 0 ) = p(1,1 |0) = p(0,1|1) = p( 1,1|1) = 1/2 $ and all other $p(a,b|x)$ be 0. Then $I_{222} = 1/2$ and $ACE= p_B(1|1) - p_B(1|0) = 0$. 

With detection efficiency $\eta$, we have  $\textrm{ACE} = \eta(p_B(1|1) - p_B(1|0))=0$, and 
\begin{equation*}
I_{222}  = \eta^2(2\times \frac{1}{2} + \frac{1}{2} +  \frac{1}{2} +  \frac{1}{2}-1) -1.
\end{equation*}
The violation condition $I_{222} - \textrm{ACE}\ge 0$ leads to $ 1.5 \eta^2 \ge 1$. Since $\eta$ is positive, this shows $\eta  >  81.65\%$, completing the proof.

\section{Proof of Theorem 4}

As in the previous section, we consider the following quantity, 
\begin{equation}
I_{222} = 2 p(0,0 | 0 ) + p(1,1 |0) + p(0,1|1) + p( 1,1|1) - 1 -\eta^2.
\end{equation}
From the previous section, we know that for classical correlations, we have 
\begin{equation}
p_B(1|1) - p_B(1|0) \ge I_{222}.
\end{equation}

We use the same quantum correlation as the one used in Sec.~\ref{AppSec:1}. 
The only difference is that here we have
\begin{equation}
\label{eq:CACEnew}
CACE^*=2 \langle M_0^0 \otimes N_0^0 \rangle + \langle M_1^0 \otimes N_1^1 \rangle + \langle M_0^1 \otimes N_1^0 \rangle + \langle M_1^1 \otimes N_1^1 \rangle - 1 -\eta^2.
\end{equation}
Substituting 
\begin{eqnarray}
\ket{\psi} &=& \cos(\alpha)  \ket{0, 0} + \sin(\alpha) \ket{1,1},  \nonumber \\
M_0^x &=&  (\mathbbm{1} +  \sin(\theta_x) \sigma_x + \cos(\theta_x) \sigma_z)/2,  \nonumber\\
M_1^x &=&   (\mathbbm{1}- \sin(\theta_x) \sigma_x-  \cos(\theta_x)\sigma_z)/2,  \nonumber\\
N_0^a &=&  (\mathbbm{1} +   \sin(\phi_a) \sigma_x +  \cos(\phi_a)\sigma_z)/2,  \nonumber \\
N_1^a  &=&  (\mathbbm{1} -  \sin(\phi_a) \sigma_x-  \cos(\phi_a) \sigma_z)/2,  \nonumber
\end{eqnarray}
into Eq.~\eqref{eq:CACEnew} with parameters $\alpha= 0.7165, \phi_0= 0.6750, \phi_1=-\phi_0,  \theta_0=0.2447,  \theta_1 = -\pi/2$,
we get $CACE^*=0.1716$. The expression of QACE is unchanged, so we still have $QACE=0$.
With imperfect detection efficiency $\eta$, we have
$CACE^*_\eta=1.1716\eta^2-1$ and $QACE_\eta=0$.
Hence, for any $\eta > 92.39\%$, the quantum correlation violates the classical causality bound, which completes the proof.

\section{Proof of Theorem 5}
We first show that, for any $\eta> 0$, there exists a no-signaling correlation that violates
\begin{equation}
\label{eq:nosignalMM}
p_B(1|1) - p_B(1|0) \ge I_{M22}.
\end{equation}
Consider the following Bell setting: 
\begin{eqnarray}
p_{Bell}(a,1-a|1,0)=1/2, &\quad \forall a=0,1; \nonumber \\
p_{Bell}(a,a|1,0)=0, &\quad  \forall a=0,1;\nonumber \\
p_{Bell}(a,1-a|x,y)=0, &\quad\forall (x,y)\not = (1,0), a=0,1;\nonumber \\
p_{Bell}(a,a|1,0)=1/2,& \quad \forall (x,y)\not = (1,0), a=0,1.\nonumber 
\end{eqnarray}
 An illustration of this no-signaling distribution for $M=3$ is shown in Fig.~\ref{fig:nosignalM22}.
 \begin{figure}[htb]
\centering \includegraphics[width=7cm]{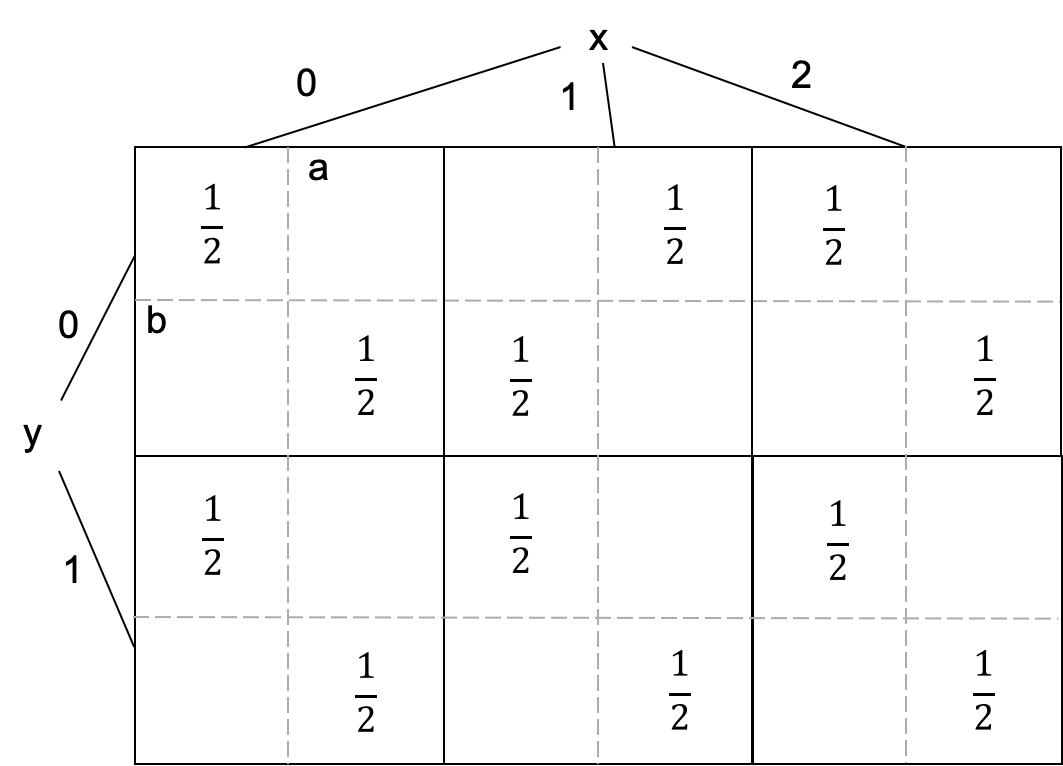}
\caption{Illustration of the no-signaling probability distribution for multiple measurement settings.} 
\label{fig:nosignalM22}
\end{figure} 
 
Clearly the no-signaling distribution satisfies 
\begin{equation}
p_{Bell} (a| x,y) = p_{Bell} (b| x,y) = 1/2,
\end{equation}
for all $x,y,a,b$. Hence, it obeys the no-signaling condition.
Let $p(a,b|x)= p_{Bell}(a,b|x,a)$ and in addition note that 
\begin{eqnarray}
p^{N}_\eta (a,b |x) &=& \eta^2 p^{N} (a,b|x),   \\
p^{N}_\eta (a | x )&=& \eta p^{N} (a | x) , \nonumber
\end{eqnarray}
where $p^N$ denotes a no-signaling probability distribution with perfect detection efficiency and $p^N_\eta$ denotes a no-signaling probability distribution with detection efficiency $\eta$.
Hence the RHS of Eq.~\eqref{eq:nosignalMM} can be written as
\begin{equation*}
\eta^2 ( \frac{M}{M-1} p(0,0| 0) + \frac{1}{M-1} p(0,1|1)  \\
 +\frac{1}{M-1} \sum\limits_{x=0}^{M-1} p(1,1|x) -\frac{1}{M-1}  ) - 1.
\end{equation*}
Since $p(b |do(a) ) =1/2$ for all $a,b$, we have that the LHS of  Eq.~\eqref{eq:nosignalMM} is 0.
Hence, Eq.~\eqref{eq:nosignalMM} is violated as long as 
\begin{equation}
\eta^2  \frac{2M-1}{2(M-1)} > 1,
 \end{equation}
which gives $\eta > \sqrt{(2M-2)/(2M-1)}$. Hence, we have proved that detection efficiency $ \sqrt{(2M-2)/(2M-1)}$
suffices for no-signaling theory to violate the inequality Eq.~\eqref{eq:nosignalMM}.

It remains to show that all classical correlations obey the inequality Eq.~\eqref{eq:nosignalMM}. 
Due to the linearity of Eq.~\eqref{eq:nosignalMM}, it suffices to consider deterministic strategies. For 
a deterministic strategy, there exist a function $F_A(x)=a$ and a function $F_B(a) =b$.

We now divide into nine cases: (1) $F_B(1)=1$ and $F_B(0)=0$; (2) $F_B(1)=0$ and $F_B(0) = 0$;  (3) $F_B(1) = 0$ and $F_B(0)=1$; (4) $F_B(1) =1$ and $F_B(0)=1$; (5) $F_B(1) = \Phi$ and $F_B(0)= \Phi$; (6) $F_B(1) = \Phi$ and $F_B(0)=0$; (7) $F_B(1) = \Phi$ and $F_B(0)=1$; (8) $F_B(1) =0$ and $F_B(0)= \Phi$; (9) $F_B(1) =1$ and $F_B(0)= \Phi$. 

{\it Case 1.}  $F_B(1)=1$ and $F_B(0)=0$. In this case, the RHS of Eq.~\eqref{eq:CACE2} is reduced to
\begin{equation}
 \frac{M}{M-1} p(0,0| 0)  + \frac{1}{M-1} \sum\limits_{x=0}^{M-1} p(1,1|x) -1 -\frac{\eta^2}{M-1} ,
\end{equation}
and the left hand side (LHS) of Eq.~\eqref{eq:CACE2} is reduced to 1.
Clearly, $F_A(0)=0$ and $F_A(1)=1$ maximize the RHS of Eq.~\eqref{eq:CACE2}, which becomes
\begin{equation}
 \frac{M}{M-1} p(0,0| 0)  + \frac{1}{M-1} \sum\limits_{x=1}^{M-1} p(1,1|x)- 1 -\frac{\eta^2}{M-1} \le  p(0,0 | 0 ) \le 1.
\end{equation}
As the maximum of the RHS of Eq.~\eqref{eq:CACE2}  is no larger than the LHS of Eq.~\eqref{eq:CACE2}, Eq.~\eqref{eq:CACE2} always holds in Case 1.

{\it Case 2.}  $F_B(1)=0$ and $F_B(0) = 0$.  In this case, the RHS of Eq.~\eqref{eq:CACE2} is reduced to
\begin{equation}
 \frac{M}{M-1} p(0,0| 0)  -1 -\frac{\eta^2}{M-1} \le 0 ,
\end{equation}
and the LHS of Eq.~\eqref{eq:CACE2} is reduced to 0.
As the maximum of the RHS of Eq.~\eqref{eq:CACE2}  is no larger than the LHS of Eq.~\eqref{eq:CACE2}, Eq.~\eqref{eq:CACE2} always holds in Case 2.

{\it Case 3.}  $F_B(1)=0$ and $F_B(0) = 1$.  In this case, the RHS of Eq.~\eqref{eq:CACE2} is reduced to
\begin{equation}
\frac{1}{M-1} p(0,1|1) -1 -\frac{\eta^2}{M-1} \le -1,
\end{equation}
and the LHS of Eq.~\eqref{eq:CACE2} is reduced to $-1$.
As the maximum of the RHS of Eq.~\eqref{eq:CACE2}  is no larger than the LHS of Eq.~\eqref{eq:CACE2}, Eq.~\eqref{eq:CACE2} always holds in Case 3.

{\it Case 4.}  $F_B(1)=1$ and $F_B(0) = 1$.  In this case, the RHS of Eq.~\eqref{eq:CACE2} is reduced to
\begin{eqnarray}
&& \frac{1}{M-1} p(0,1|1) + \frac{1}{M-1} \sum\limits_{x=0}^{M-1} p(1,1|x) -1 -\frac{\eta^2}{M-1}  \nonumber \\ 
&\le&  \frac{1}{M-1}p(1,1 |0) + \frac{1}{M-1} [p(0,1|1) + p( 1,1|1)]   + \frac{1}{M-1} \sum\limits_{x=2}^{M-1} p(1,1|x) - 1 -\frac{\eta^2}{M-1}  \\  
&\le&  \frac{M}{M-1}\eta^2  -\frac{\eta^2}{M-1}\le0, \nonumber
\end{eqnarray}
and the LHS of Eq.~\eqref{eq:CACE2} is reduced to 0.
As the maximum of the RHS of Eq.~\eqref{eq:CACE2}  is no larger than the LHS of Eq.~\eqref{eq:CACE2}, Eq.~\eqref{eq:CACE2} always holds in Case 4.
 
{\it Case 5.}   $F_B(1) = \Phi$ and $F_B(0)= \Phi$. In this case, the RHS of Eq.~\eqref{eq:CACE2} is reduced to
$- 1 -\eta^2/(M-1)$,
and the LHS of Eq.~\eqref{eq:CACE2} is reduced to 0.
As the RHS of Eq.~\eqref{eq:CACE2}  is no larger than the LHS of Eq.~\eqref{eq:CACE2}, Eq.~\eqref{eq:CACE2} always holds in Case 5.

{\it Case 6.}  $F_B(1) = \Phi$ and $F_B(0)=0$. 
In this case, the RHS of Eq.~\eqref{eq:CACE2} is reduced to
\begin{equation}
 \frac{M}{M-1} p(0,0| 0)  -1 -\frac{\eta^2}{M-1}\le 0 ,
\end{equation}
and the LHS of Eq.~\eqref{eq:CACE2} is reduced to 0.
As the maximum of the RHS of Eq.~\eqref{eq:CACE2}  is no larger than the LHS of Eq.~\eqref{eq:CACE2}, Eq.~\eqref{eq:CACE2} always holds in Case 6.

{\it Case 7.}  $F_B(1) = \Phi$ and $F_B(0)=1$. 
In this case, the RHS of Eq.~\eqref{eq:CACE2} is reduced to
\begin{equation}
 \frac{1}{M-1} p(0,1|1)  -1 -\frac{\eta^2}{M-1} \le -1 ,
\end{equation}
and the LHS of Eq.~\eqref{eq:CACE2} is reduced to $-1$.
As the maximum of the RHS of Eq.~\eqref{eq:CACE2}  is no larger than the LHS of Eq.~\eqref{eq:CACE2}, Eq.~\eqref{eq:CACE2} always holds in Case 7.

{\it Case 8.}  $F_B(1) =0$ and $F_B(0)= \Phi$. In this case, the RHS of Eq.~\eqref{eq:CACE2} is reduced to
$- 1 -\eta^2/(M-1)$,
and the LHS of Eq.~\eqref{eq:CACE2} is reduced to 0.
As the RHS of Eq.~\eqref{eq:CACE2}  is no larger than the LHS of Eq.~\eqref{eq:CACE2}, Eq.~\eqref{eq:CACE2} always holds in Case 8.

{\it Case 9.}  $F_B(1) =1$ and $F_B(0)= \Phi$. 
In this case, the RHS of Eq.~\eqref{eq:CACE2} is reduced to
\begin{equation}
 \frac{1}{M-1} \sum\limits_{x=0}^{M-1} p(1,1|x) -1 -\frac{\eta^2}{M-1} \le 0,
\end{equation}
and the LHS of Eq.~\eqref{eq:CACE2} is reduced to $1$.
As the maximum of the RHS of Eq.~\eqref{eq:CACE2}  is no larger than the LHS of Eq.~\eqref{eq:CACE2}, Eq.~\eqref{eq:CACE2} always holds in Case 9.

\section{Proof of Theorem 6}
First we show that the quantum causal bound Eq.~(12) in the main text is tight for quantum correlations.
Consider the following quantum correlation:
\begin{eqnarray}
p(1,0 |x) =0, &\quad \forall x=0,1, \nonumber \\
p(1,1 |x) =1/2,& \quad \forall x=0,1, \nonumber \\
p(0,1 |x) =\frac{1}{2}-p(0,0|x), &  \quad \forall x=0,1,  \\
p(0,0 |0) = \sqrt{2}-1, &\nonumber \\
p(0,0 |1) =0. & \nonumber 
\end{eqnarray}
An illustration of this probability distribution is shown in Fig.~\ref{fig:QBQC}.
 \begin{figure}[htb]
\centering \includegraphics[width=5cm]{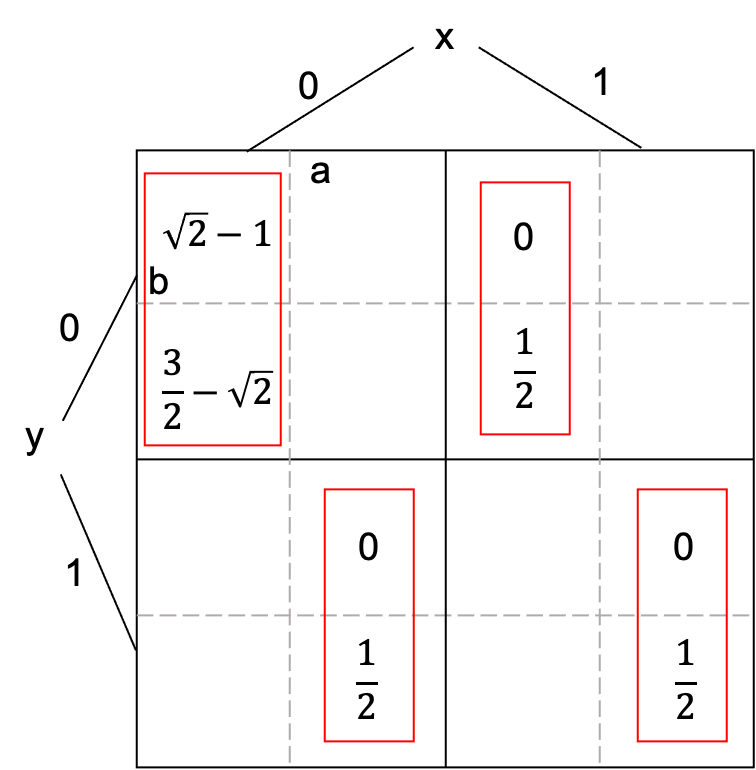}
\caption{Illustration of the quantum probabilistic distribution which maximally saturates the quantum causal bound.} 
\label{fig:QBQC}
\end{figure} 

Then we have
\begin{eqnarray}
\sum\limits_{x=0,1} (-1)^x[ p(0,0|x) -  p(0,1|x) ] &=& 2\sqrt{2} -2 , \nonumber \\
\sum\limits_{x=0,1} (-1)^x[ p(1,0|x) -  p(1,1|x) ] &=&0,
\end{eqnarray}
which leads to 
\begin{equation}
\xi = \sqrt{ [1-(2\sqrt{2}-2)]  (1-0)} = \sqrt{2} -1.
\end{equation}
Therefore the RHS of Eq.~(12) in the main text becomes 
\begin{equation}
(\sqrt{2}-1)+\frac{1}{2}+0+\frac{1}{2} - (\sqrt{2}-1)-1=0.
\end{equation}
This is the borderline that the LHS of Eq.~(12) in the main text equals 0 according to Fig.~2 of Ref.~\cite{PhysRevLett.125.230401}. Hence, the inequality Eq.~(12) in the main text attains equality, which shows that
Eq.~(12) in the main text is a tight causal bound for quantum correlations. 

Next, we consider the no-signaling correlation that has been illustrated in Fig.~\ref{fig:nosignal222}, namely
\begin{eqnarray}
p(a,1-a |x) =1/2, &\quad (x,y)=(1,0), \forall a=0,1 \nonumber \\
p(a,a |x) =1/2, &\quad  \forall (x,y)\not=(1,0), a=0,1,  \\
p(a,b |x) =0,  &\quad \textrm{otherwise}. \nonumber
\end{eqnarray}
With this correlation, we have
\begin{eqnarray}
\sum\limits_{x=0,1} (-1)^x[ p(0,0|x) -  p(0,1|x) ] &=& 1 , \nonumber \\
\sum\limits_{x=0,1} (-1)^x[ p(1,0|x) -  p(1,1|x) ] &=&0,
\end{eqnarray}
which leads to 
\begin{equation}
\xi = \sqrt{ [1-1]  (1-0)} =0.
\end{equation}
Therefore the RHS of Eq.~(12) in the main text, denoted by QACE$^*$ becomes 
\begin{equation}
\frac{1}{2}+\frac{1}{2}+0+\frac{1}{2} -0-1=\frac{1}{2}.
\end{equation}
The LHS of Eq.~(12) in the main text, denoted by NACE, is 0.
The no-signaling correlation then achieves the maximum violation of the quantum causal bound by 
\begin{equation}
\textrm{QACE}^* - \textrm{NACE} = \frac{1}{2} -0 =\frac{1}{2}.
\end{equation}

With detection efficiency $\eta$, let $ \textrm{NACE}_\eta$ denote the LHS of Eq.~(12) in the main text
and  $\textrm{QACE}^*_\eta$ denote the RHS of Eq.~(12) in the main text. Then we have
\begin{equation}
\textrm{NACE}_\eta = \eta \textrm{NACE} = 0.
\end{equation}
With detection efficiency $\eta$, $\xi$ becomes 
\begin{equation}
\xi_\eta = \sqrt{ [1-\eta^2]  (1-0)} =\sqrt{1-\eta^2},
\end{equation}
and hence
\begin{eqnarray}
\textrm{QACE}^*_\eta &=& \eta^2(\frac{1}{2}+\frac{1}{2}+0+\frac{1}{2})  -\xi_\eta -1  \nonumber \\
&=&\frac{3}{2} \eta^2 -\sqrt{1-\eta^2} -1. 
\end{eqnarray}
A no-signaling violation of the quantum causal bound means that
\begin{eqnarray}
0 &<& \textrm{QACE}^*_\eta-  \textrm{NACE}_\eta   \\
   &=& \frac{3}{2} \eta^2 -\sqrt{1-\eta^2} -1,        \nonumber
\end{eqnarray}
which gives $\eta > 94.29\%$. This completes the proof.

\section{Proof of Theorem 7}

First, we prove the following lemma.
\begin{lemma}
\label{lemma:max}
\begin{equation}
\max\limits_\alpha(- \alpha x'  - \beta y' - |  \alpha + \beta |) =\frac{1}{2}(x'-y') -\min\limits_\pm \sqrt{(1\pm x')(1\pm y')}|.
\end{equation}
\end{lemma}
\begin{proof}
Consider two cases: (1) $ |\alpha + \beta | \ge 0$; (2) $ |  \alpha + \beta |  < 0$.

{\it Case 1:}  Combined with the relation
\begin{equation}
\beta = \frac{1+ \alpha}{ 1+ 2\alpha}, 
\end{equation}
then the quantity to maximize in the LHS of Eq.~\eqref{eq:max} becomes
\begin{eqnarray}
&&  - \alpha x'  - \beta y' -   \alpha  - \beta  \nonumber   \\
&=&   - \alpha (x'+1)  -  \frac{1+ \alpha}{ 1+ 2\alpha} (y'+1)    \\
&=&  -  \frac{1}{  2}(2\alpha+1) (x'+1)  + \frac{1}{  2}  (x'+1)  -  \frac{1}{  2}  (y'+1)  -  \frac{1}{2( 1+ 2\alpha)}    (y'+1)  \nonumber  \\
&\le &  \frac{1}{  2}  (x'-y') -\sqrt{(1+x')(1+y')}. \nonumber
\end{eqnarray}

{\it Case 2:}  The quantity to maximize in the LHS of Eq.~\eqref{eq:max} becomes
\begin{eqnarray}
&&  - \alpha x'  - \beta y' +   \alpha  + \beta  \nonumber   \\
&=&   - \alpha (x'-1)  -  \frac{1+ \alpha}{ 1+ 2\alpha} (y'-1)    \\
&=&  -  \frac{1}{  2}(2\alpha+1) (x'-1)  + \frac{1}{  2}  (x'-1)  -  \frac{1}{  2}  (y'-1)  -  \frac{1}{2( 1+ 2\alpha)}   (y'-1)  \nonumber   \\
&\le &  \frac{1}{  2}  (x'-y') -\sqrt{(1-x')(1-y')}.  \nonumber
\end{eqnarray}

Combining Case 1 and Case 2, the lemma is proved.
\end{proof}

Now we are ready to prove the theorem.
First, we show that quantum correlations obey Eq.~(13) in the main text.
 With detection efficiency $\eta$, we have the relations $\sum_{a,b=0}^1 p(a,b|x) = \eta^2, \forall x$.
 Let 
\begin{equation}
B(\alpha, \beta, \gamma, \delta) =  -\alpha  \langle M^0 \otimes N^0 \rangle + \beta \langle M^0 \otimes N^1 \rangle + \gamma \langle M^1 \otimes N^0 \rangle + \delta \langle M^1 \otimes N^1 \rangle,
\end{equation}
where $M^x = M^x_0 - M^x_1$,  $N^y = N^y_0 - N^y_1$, and $ \langle M^x \otimes N^y \rangle = \textrm{tr} [ (M^x \otimes N^y ) \rho_{AB} ]  $.
By the relations
$\langle M^x_a \otimes N^a_b\rangle = p(a,b|x)$ and $\langle \mathbbm{1} \otimes N^a_b\rangle = p(b| do(a))$, 
we have
\begin{eqnarray}
B(\alpha, \beta, \gamma, \delta) &=& -\eta^2(\alpha - \beta +\gamma - \delta) + 2\eta p(0| do(0)) (\alpha -\gamma) + 2 \eta p( 0|do(1) ) (\beta + \delta) \\
 &&- 2\alpha(p(0,0|0) - p(0,1|0)) + 2\beta( p(1,1|0) - p(1,0|0) )  \nonumber  \\
& &+ 2\gamma (p(0,0|1)- p(0,1|1))  + 2\delta( p(1,1|1) - p(1,0|1) ).   \nonumber 
\end{eqnarray}
 Let $\gamma = 1 + \alpha$ and $\delta = 1 - \beta$. After reorganizing the terms, we have 
\begin{eqnarray}
\label{eq:QACE2inter}
\eta QACE &\ge& \eta [p(0|do(0)) - p(0|do(1))]   \\
&= & p(0,0|1) + p(1,1|1) - p(0,1|1)  - p(1,0|1)   - \alpha \sum\limits_{x=0,1} (-1)^x ( p(0,0|x)  - p(0,1|x) )  \nonumber  \\
 && - \beta \sum\limits_{x=0,1} (-1)^x ( p(1,0|x)  - p(1,1|x) ) - B(\alpha, \beta, \gamma, \delta).\nonumber
\end{eqnarray}

Equality~\eqref{eq:QACE2inter} holds for any $\alpha$ and $\beta$. We next  select appropiate $\alpha$ and $\beta$ such that the RHS of Eq.~\eqref{eq:QACE2inter} is maximized. Let us set $\beta = (1+ \alpha)/(1+2\alpha)$,  and we have
$B(\alpha, \beta, \gamma, \delta) = | \alpha + \beta |$. In addition, 
\begin{eqnarray}
\label{eq:max}
&& \max\limits_\alpha [- \alpha \sum\limits_{x=0,1} (-1)^x ( p(0,0|x)  - p(0,1|x) ) - \beta \sum\limits_{x=0,1} (-1)^x ( p(1,0|x)    - p(1,1|x) ) - |  \alpha + \beta |] \nonumber  \\
&=& \frac{1}{2} [  \sum\limits_{x=0,1} (-1)^x ( p(0,0|x)  - p(0,1|x) ) - \sum\limits_{x=0,1} (-1)^x ( p(1,0|x)   - p(1,1|x) )  ] - \xi,  
 \end{eqnarray}
 where
 \begin{equation}
\xi = \min\limits_{\pm} \sqrt{  \prod\limits_{a=0,1} \{ 1\pm \sum\limits_{x=0,1} (-1)^x[ p(a,0|x) -  p(a,1|x) ]  \}   } .
 \end{equation}
This equation is obtained by substituting 
\begin{eqnarray}
x' = \sum\limits_{x=0,1} (-1)^x ( p(0,0|x)  - p(0,1|x) ),  \\
y'=  \sum\limits_{x=0,1} (-1)^x ( p(1,0|x)  - p(1,1|x) ).
\end{eqnarray}
into Lemma~\ref{lemma:max}.

Note the following identity
 \begin{eqnarray}
 \label{eq:probmani}
&& p(0,0|1) +p(1,1|1) -p(0,1|1) - p(1,0|1)  +\frac{1}{2} [  \sum\limits_{x=0,1} (-1)^x ( p(0,0|x)  - p(0,1|x) )  - \sum\limits_{x=0,1} (-1)^x ( p(1,0|x)  - p(1,1|x) )  ]   \nonumber\\
&=& p(0,0|1) +p(1,1|1) -p(0,1|1) - p(1,0|1)  +\frac{1}{2} [ p(0,0|0) - p(0,1|0) \nonumber  \\
&&\quad- p(0,0|1) + p(0,1|1) -p(1,0|0)+p(1,1|0) + p(1,0|1) - p(1,1|1)  ]   \\
&=& -p(0,1|1) - p(1,0|1)  +\frac{1}{2} [ p(0,0|0) - p(0,1|0) + p(0,0|1) + p(0,1|1) -p(1,0|0)+p(1,1|0) + p(1,0|1) + p(1,1|1)  ] \nonumber \\
&=& -p(0,1|1) - p(1,0|1)  + \frac{1}{2} [ p(0,0|0) - p(0,1|0)-p(1,0|0)+p(1,1|0)  ]  + \frac{1}{2} [ p(1,0|1) + p(1,1|1) + p(0,0|1) + p(0,1|1) ] \nonumber  \\
&=& -p(0,1|1) - p(1,0|1)  + \frac{1}{2} [ 2p(0,0|0) +2p(1,1|0) -\eta^2  ] + \frac{1}{2} [ \eta^2 ] \nonumber \\
&=& p( 0,0|0 ) + p( 1,1|0) -p(0,1|1) - p(1,0|1).  \nonumber
\end{eqnarray}
 Combining Eqs.~\eqref{eq:QACE2inter}, \eqref{eq:max}, and \eqref{eq:probmani}, we have
\begin{equation}
\eta QACE \ge p( 0,0|0 ) + p( 1,1|0) -p(0,1|1) - p(1,0|1) - \xi,
\end{equation}
which shows that Eq.~(13) in the main text holds.

Next, we consider the no-signaling correlation that has been illustrated in Fig.~\ref{fig:nosignal222}, namely
\begin{eqnarray}
p(a,1-a |x) =1/2, &\quad (x,y)=(1,0), \forall a=0,1 \nonumber \\
p(a,a |x) =1/2, &\quad  \forall (x,y)\not=(1,0), a=0,1,  \\
p(a,b |x) =0,  &\quad \textrm{otherwise}. \nonumber
\end{eqnarray}
With this correlation, we have
\begin{eqnarray}
\sum\limits_{x=0,1} (-1)^x[ p(0,0|x) -  p(0,1|x) ] = 1 , \nonumber \\
\sum\limits_{x=0,1} (-1)^x[ p(1,0|x) -  p(1,1|x) ] =0,
\end{eqnarray}
which leads to 
\begin{equation}
\xi = \sqrt{ [1-1]  (1-0)} =0.
\end{equation}
Therefore the RHS of Eq.~(13) in the main text, denoted by QACE$^*$ becomes 
\begin{equation}
\frac{1}{2}+\frac{1}{2}-\frac{1}{2}-0 -0=\frac{1}{2}.
\end{equation}
The LHS of Eq.~(13) in the main text, denoted by NACE, is 0.
The no-signaling correlation then achieves the maximum violation of the quantum causal bound by 
\begin{equation}
\textrm{QACE}^* - \textrm{NACE} = \frac{1}{2} -0 =\frac{1}{2}.
\end{equation}

With detection efficiency $\eta$, let $ \textrm{NACE}_\eta$ denote the LHS of Eq.~(13) in the main text
and  $\textrm{QACE}^*_\eta$ denote the RHS of Eq.~(13) in the main text. Then we have
\begin{equation}
\textrm{NACE}_\eta = \eta \textrm{NACE} = 0.
\end{equation}
With detection efficiency $\eta$, $\xi$ becomes 
\begin{equation}
\xi_\eta = \sqrt{ (1-\eta^2)  (1-0)} =\sqrt{1-\eta^2},
\end{equation}
and hence
\begin{eqnarray}
\textrm{QACE}^*_\eta &=&\frac{1}{\eta} [\eta^2(\frac{1}{2}+\frac{1}{2}-\frac{1}{2}-0)  -\xi_\eta ]  \nonumber \\
&=& \frac{1}{\eta}(\frac{1}{2} \eta^2 -\sqrt{1-\eta^2}). 
\end{eqnarray}
A no-signaling violation of the quantum causal bound implies
\begin{eqnarray}
0 &<& \textrm{QACE}^*_\eta-  \textrm{NACE}_\eta   \\
   &=&  \frac{1}{\eta}(\frac{1}{2} \eta^2 -\sqrt{1-\eta^2}),        \nonumber
\end{eqnarray}
which gives $\eta > 91.02\%$. This completes the proof.

\end{document}